\documentclass[lineno]{biometrika}

\usepackage[acronym,toc]{glossaries}
\makeglossaries
\usepackage{amsmath,commath,amssymb,mathtools,multicol,caption,subcaption,color}
\usepackage[bbgreekl]{mathbbol}

\newcommand{\vmat}[1]{\begin{vmatrix}#1\end{vmatrix}}

\newcommand{\mat}[1]{\begin{bmatrix}#1\end{bmatrix}}

\DeclareMathOperator{\pr}{\mathsf{p}}
\DeclarePairedDelimiterX{\SqrBrk}[1]{[}{]}{#1}
\newcommand{\p}[1]{\pr_{#1}\SqrBrk*}
\newcommand{\phat}[1]{\widehat{\pr}_{#1}\SqrBrk*}
\newcommand{\widebar}[1]{\overline{#1}}
\DeclareMathAlphabet{\mathpzc}{OT1}{pzc}{m}{it}
\newcommand{\scripted}[1]{\mathpzc{#1}}

\newcommand{\MN}[1]{\min\left\lbrace #1 \right\rbrace}

\DeclareMathOperator{\Exp}{\mathsf{E}}
\newcommand{\E}[1]{\Exp_{#1}\SqrBrk*}
\DeclareMathOperator{\Var}{\mathsf{var}}
\newcommand{\V}[1]{\Var_{#1}\SqrBrk*}
\DeclareMathOperator{\Cor}{\mathsf{cor}}
\newcommand{\C}[1]{\Cor_{#1}\SqrBrk*}
\newcommand{\defeq}{\triangleq}
\newcommand{\almostsure}{\xrightarrow{a.s.}}
\newcommand{\distconv}{\xrightarrow{d}}
\usepackage{times}
\usepackage{bm}
\usepackage{natbib}

\usepackage[plain,noend]{algorithm2e}

\makeatletter
\renewcommand{\algocf@captiontext}[2]{#1\algocf@typo. \AlCapFnt{}#2} % text of caption
\def\@algocf@capt@plain{top}
\renewcommand{\algocf@makecaption}[2]{%
  \addtolength{\hsize}{\algomargin}%
  \sbox\@tempboxa{\algocf@captiontext{#1}{#2}}%
  \ifdim\wd\@tempboxa >\hsize%     % if caption is longer than a line
    \hskip .5\algomargin%
    \parbox[t]{\hsize}{\algocf@captiontext{#1}{#2}}% then caption is not centered
  \else%
    \global\@minipagefalse%
    \hbox to\hsize{\box\@tempboxa}% else caption is centered
  \fi%
  \addtolength{\hsize}{-\algomargin}%
}
\makeatother

\begin{document}
\jname{}
%% The year, volume, and number are determined on publication
\jyear{}
\jvol{}
\jnum{}
%% The left and right page headers are defined here:
\markboth{Khoa T. Tran}{Generalised Metropolis--Hastings and other Markov chain Monte Carlo Algorithms}

%% Here are the title, author names and addresses
\title{A Common Derivation for Markov Chain Monte Carlo Algorithms with Tractable and Intractable Targets}

\author{Khoa T. Tran}
\affil{}

\maketitle

\begin{abstract}
Markov chain Monte Carlo is a class of algorithms for drawing Markovian samples from high dimensional target densities to approximate the numerical integration associated with computing statistical expectation, especially in Bayesian statistics. However, many Markov chain Monte Carlo algorithms do not seem to share the same theoretical support and each algorithm is proven in a different way. This incurs a large amount of terminologies and ancillary concepts, which makes Markov chain Monte Carlo literature seems to be scattered and intimidating to researchers from many other fields, including new researchers of Bayesian statistics.

A generalised version of the Metropolis--Hastings algorithm is constructed with a random number generator and a self--reverse mapping. This formulation admits many other Markov chain Monte Carlo algorithms as special cases. A common derivation for many Markov chain Monte Carlo algorithms is useful in drawing connections and comparisons between these algorithms. As a result, we now can construct many novel combinations of multiple Markov chain Monte Carlo algorithms that amplify the efficiency of each individual algorithm. Specifically, we reinterpret slice sampling as a special case of Metropolis--Hastings and then propose two novel sampling schemes that combine slice sampling with directional or Hamiltonian sampling. Our Hamiltonian slice sampling scheme is also applicable in the pseudo marginal context where the target density is intractable but can be unbiasedly estimated, e.g. using particle filtering.
\end{abstract}


\section{Introduction}
Many statistical analyses often require the numerical evaluation of the expectation of an arbitrary function of interest, $f:\mathbb{R}^{n_\theta}\mapsto\mathbb{R}$, with respect to a given density $\pi(\theta)$ defined as
\begin{equation*}
  \E{\pi}{f}\defeq\int f(\theta)\pi(\theta)\dif\theta,\quad\theta\in\mathbb{R}^{n_\theta} ,
\end{equation*} where $\theta$ is the parametric vector in a given statistical model. This computation appears to be indeed a primal goal in many research areas, especially Bayesian statistics, because the expectation $\E{\pi}{f}$ actually encapsulates all statistical information contained in the density $\pi(\theta)$. For example, letting $f(\theta)\equiv\mathbb{1}_\scripted{A}(\theta)$ for any set $\scripted{A}\in\mathbb{R}^{n_\theta}$ will lead to $\E{\pi}{f} = \wp(A)$, which is the measure of $\scripted{A}$ with respect to the law of distribution given by $\pi(\theta)$. Other choices of the function $f(\theta)$ also lead to an even more flexible computational framework. When we are also interested in the expectation of another function $g(\theta)$ with respect to another density $\phi(\theta)$, then we can simply define $f(\theta)\defeq g(\theta)\phi(\theta)/\pi(\theta)$ and compute $\E{\phi}{g} = \E{\pi}{f}$. Furthermore, if we only have access to an unnormalised version of $\pi(\theta)$ denoted as $\widetilde{\pi}(\theta)\defeq\pi(\theta) Z_{\pi}$, then the normalising constant $Z_\pi$, which can be of interest in itself according to e.g. \cite{Skilling2006a,Feroz2013}, can be computed as $Z_\pi{=}1/\E{\pi}{\phi(\theta)/\widetilde{\pi}(\theta)}$ with $\phi(\theta)$ being some standard density such as a Gaussian.

Given such a central role and computational flexibility of the integration $\E{\pi}{f}$, there is great utility in designing general--purpose algorithms for approximating this expectation for arbitrary function $f(\theta)$ and density $\pi(\theta)$. However, the expectation $\E{\pi}{f}$ is also often a high dimensional integration problem, which renders conventional numeric approximation techniques such as Simpson’s rule ineffective due to e.g. unknown integration boundaries, exponential growth of the computational load with respect to the number of dimensions \citep{Kuo2005}.
   
Monte Carlo integration helps to reduce the computational load of high dimensional integration by way of random sampling from the target density $\theta_k\sim\pi(\cdot), k = 1,2,\dots M,$ and relying on the strong law of large number \citep{Gut2013} to approximate the required integral as follows
	\begin{equation*}
	\widehat{f}_M\defeq\frac{1}{M}\sum_{k=1}^M f(\theta_k) \almostsure\E{\pi}{f}.
	\end{equation*}	
   
One successful approach in sampling from high dimensional arbitrary density is to simulate a Markov chain with transition kernel $\scripted{K}(\dif\theta_{k+1}\mid\theta_k) $ that has an invariant density equal to the target density $\pi(\theta)$, i.e. $\theta_k\sim\pi(\cdot)\Rightarrow\theta_{k+1}\sim\pi(\cdot)$. According to the law of total probability, this invariance condition means $\scripted{K}(\dif\theta_{k+1}\mid\theta_k) $ has to at least satisfy the following  equality
\begin{equation}\label{eq:invariant}
	\int\scripted{K}(\dif\theta_{k+1}\mid\theta_k)\wp(\dif\theta_k) = \wp(\dif\theta_{k+1});\quad\wp(\dif\theta)\defeq\pi(\theta)\dif\theta .
\end{equation}

A central limit theorem is given by \cite{Chan1994} for Markovian sampling as follows
	\begin{equation}\label{eq:convergence}
	\begin{aligned}
		\sqrt{M}\left(\widehat{f}_M - \E{\pi}{f}\right) &\distconv\mathcal{N}(0,\sigma^2_{\widehat{f}}),\\
		\sigma^2_{\widehat{f}} =\V{\pi}{f(\theta)}\tau_{\widehat{f}},\quad
		\tau_{\widehat{f}} &\defeq  1+2\sum_{k=2}^\infty\C{}{f(\theta_1),f(\theta_k)}.
	\end{aligned}
	\end{equation}
The so called integrated autocorrelation time $\tau_{\widehat{f}}$ given in the central limit theorem (\ref{eq:convergence}) is a popular measure of inefficiency in using the Markovian samples $\theta_k,k=1,2\dots,M$ to estimate $\E{\pi}{f}$, since it is the ratio between the estimator variance and the true variance of $f(\theta)$ with respect to $\pi(\theta)$. This factor also highlights the fact that strongly correlated samples will lead to large estimator variance, i.e. inaccurate estimation of $\E{\pi}{f}$. Therefore the central goal in designing Markov chain Monte Carlo algorithms is to deliver Markovian samples $\theta_k\sim\pi(\cdot)$ that are as uncorrelated as possible. In this regard, some Markov chain Monte Carlo algorithms can outperform others depending on the specific application. However, choosing the best algorithm for a given application among many available sampling schemes is often a nontrivial matter.

The Metropolis--Hastings algorithm, which was pioneered by \cite{Metropolis1953} and later generalised by \cite{Hastings1970}, is a fairly general framework to construct the required transition kernel $\scripted{K}(\dif\theta_{k+1}\mid\theta_k)$ for arbitrary target densities $\pi(\theta)$ at moderate dimensionality. However, there are many other algorithms for high dimensional applications, such as Gibbs \citep{Geman1984}, Hamiltonian \citep{Duane1987}, directional \citep{Gilks1994}, univariate or elliptical slice  sampling \citep{Neal2003,Murray2010}, that appear on first reading to not fit into the current framework of Metropolis--Hastings as described by \cite{Roberts2004a}.

The independent understanding of many algorithms presents a few challenges to the readers of Bayesian statistics literature. First, it is difficult and also time consuming to understand multiple Markov chain Monte Carlo algorithms at once. It can also be much harder to compare or make any theoretical connections between multiple Markov chain Monte Carlo algorithms. Finally, it is difficult to judge when to use which algorithm among a plethora of options.

This work aims to address these difficulties by constructing a common mathematical framework for understanding multiple Markov chain Monte Carlo algorithms at once, which allows both transparent comparisons between these algorithms, and advantageous combinations that amplifies the efficiency of each individual algorithm. We now present this framework in the following section before deriving other Markov chain Monte Carlo algorithms as special cases.

\section{Generalised Metropolis--Hastings}
\subsection{The Current Framework}
The conventional Metropolis--Hastings algorithm is essentially constructed from a base proposal density $\gamma(\xi\mid\theta),\xi\in\mathbb{R}^{n_\theta},$ and an acceptance test which is formulated in a way that ensures the correct invariant density of the resulting chain is identical to $\pi(\theta)$. Specifically, the potential candidate for the next sample is generated from a tractable density in $\mathbb{R}^{n_\theta}$ denoted as
\begin{equation}\label{eq:MH_proposal}
  \xi\sim\gamma(\xi\mid\theta),
\end{equation} where $\gamma(\xi\mid\theta)$ is constructed from elementary random number generators to approximate the target $\pi(\xi)$. The mismatch between $\gamma(\xi\mid\theta)$ versus $\pi(\xi)$ is corrected by applying an acceptance test where for some random scalar $u\sim\mathcal{U}[0,1]$ and an acceptance probability defined as
\begin{equation}\label{eq:MH_acc}
	\alpha(\xi\mid\theta)\defeq\min\left\lbrace 1,\frac{\pi(\xi)}{\pi(\theta)}\frac{\gamma(\theta\mid\xi)}{\gamma(\xi\mid\theta)} \right\rbrace,
\end{equation} we choose the new sample to be
\begin{equation}\label{eq:MH}
	\theta_{k+1} = \begin{cases}
		\xi_k &\text{if }u_k\le \alpha(\xi_k\mid\theta_k),\\
		\theta_k &\text{otherwise}. 
	\end{cases}
\end{equation}

While being a very elegant approach in Markovian sampling, the Metropolis--Hastings framework cannot be employed to explain many other sampling algorithms. One important aspect of this framework is that the proposal density $\gamma(\xi\mid\theta)$ cannot be defined properly for other samplers, in the sense that while the proposal $\xi$ can be drawn by simulation, the associated density $\gamma(\xi\mid\theta)$ cannot be computed or properly defined in the same sampling space of $\theta\in\mathbb{R}^{n_\theta}$. 

For example, the Gibbs samplers in \citep{Geman1984,Casella1992,Damien1999}, which have unity acceptance probability, employ the following proposal
\begin{equation}\label{eq:Gibbs}
		\xi^d \sim\pi(\xi^d\mid\theta^{-d});\quad\xi^{-d} = \theta^{-d},
\end{equation} where $\theta^d\in\mathbb{R}^{n_d}$ is some subspace of $\theta\in\mathbb{R}^{n_\theta}$ and $\theta^{-d}\defeq\theta\setminus\theta^d$. In trying to explain Gibbs sampling using the Metropolis--Hastings framework, we can superficially recognise that the proposal density $\gamma(\xi\mid\theta)$ in this case is identical to $\pi(\xi^d\mid\theta^{-d})$. Then since $\xi^{-d} = \theta^{-d}$ by design and $\pi(\theta) = \pi(\theta^d\mid  \theta^{-d})\pi(\theta^{-d})$, we can show that the acceptance probability is unity as follows 
\begin{equation}\label{eq:Gibbs_acc}
	\alpha(\xi\mid\theta) =
	1\wedge\frac{\pi(\xi^d\mid  \xi^{-d})\pi(\xi^{-d})}
	{\pi(\theta^d\mid \theta^{-d})\pi(\theta^{-d})} 
	\cdot \frac
	{\pi(\theta^d\mid \xi^{-d})}
	{\pi(\xi^d\mid \theta^{-d})} = 1\quad\Rightarrow\quad\theta_{k+1}\equiv\xi_k.
\end{equation}

However, this interpretation of Gibbs sampling is improper because we have ignored the fact that this proposal only has a well--defined density in some subspace of $\mathbb{R}^{n_\theta}$, while the conventional Metropolis--Hastings requires $\gamma(\xi\mid\theta)$ to be a density in $\mathbb{R}^{n_\theta}$ \citep{Roberts2004a}. A proper explanation for Gibbs sampling, isolated from the reasoning of Metropolis--Hastings sampling, can be found in \citep{Tierney1994,Chan1994}.

While the problem appears to be only cosmetic in the case of Gibbs sampling, other sampling approaches present more significant deviation from the current framework of Metropolis--Hastings. We can see this in the case of directional sampling, as described by \cite{Roberts1994,Gilks1994,Chen1996a}, which has the following type of proposal
\begin{equation}\label{eq:directional}
	\xi = \theta + \scripted{r}(\scripted{v} + \rho\theta),
\end{equation} for some random variables $\scripted{v}\in\mathbb{R}^{n_\theta};\;\scripted{r}\in\mathbb{R}$ and a constant scalar $\rho$. For any given proposal density $\scripted{q}(\scripted{r},\scripted{v}\mid\theta)$, it is not clear how to derive the analytical expression for $\gamma(\xi\mid\theta)$ so the expression for the acceptance probability (\ref{eq:MH_acc}) is to no avail here.

The Hamiltonian Monte Carlo method, which originated from physics literature \citep{Duane1987}, is also actively studied again by e.g. \cite{Girolami2011,Neal2012} to deal with high dimensional densities. This approach includes some exotic proposal mechanisms such as elliptical slice sampling by \cite{Murray2010} or No-U-Turn sampling by \cite{Hoffman2014}, which result in both high acceptance rate and low sample correlation. In this approach, the proposal $\xi$ is constructed from the solution of the Hamiltonian differential equations
\begin{equation}\label{eq:Hamiltonian}
	\begin{aligned}
		\scripted{H}(\theta,\scripted{v}) &\defeq -\log\left( \pi(\theta)\scripted{q}(\scripted{v}\mid\theta) \right);\quad\theta,\scripted{v}\in\mathbb{R}^{n_\theta}, \\
		\frac{\dif\theta^i}{\dif t} &= \frac{\partial\scripted{H}(\theta,\scripted{v})}{\partial\scripted{v}^i};\quad\frac{\dif\scripted{v}^i}{\dif t} = \frac{- \partial\scripted{H}(\theta,\scripted{v})}{\partial\theta^i};\quad i = 1,2,\dots n_\theta,
	\end{aligned}
\end{equation} for a finite random duration $t\in[0,\scripted{r}]$ such that the trajectory starts at $\theta(0) \defeq \theta$ and ends at $\theta(r)\defeq\xi$. The auxiliary variable $\scripted{v}\sim\scripted{q}(\scripted{v}\mid\theta)$ can be interpreted as a fictional random momentum vector while $\theta$ can now be interpreted as the position of a physical particle in a gravitational field characterised by $\pi(\theta)$. The solution to (\ref{eq:Hamiltonian}) then defines a trajectory $\theta(t)$ through the space $\mathbb{R}^{n_\theta}$ that is naturally guided by the gradient of $\pi(\theta)$ and also preserves the total kinetic and potential energy represented by the fictional Hamiltonian term $\scripted{H}(\theta,\scripted{v})$. This is essentially why the proposed sample $\xi$ can travel very far away from its origin at $\theta$ while maintaining a high acceptance probability, which is indeed unity if (\ref{eq:Hamiltonian}) is solved analytically. Again in this case, there is really no hope in retrieving an analytical expression for $\gamma(\xi\mid\theta)$, and hence the current Metropolis--Hastings framework cannot properly explain the Hamiltonian sampling approach.

While we can rest assured that these and other algorithms all have their own theoretical support, the technical concepts and terminologies associated with each of them can be overwhelming to most readers of Bayesian computation literature, especially applied researchers who seek to employ these computational tools in their domain specific applications. Each of these types of proposals, e.g. (\ref{eq:Gibbs}), (\ref{eq:directional}) and (\ref{eq:Hamiltonian}), leads to different ways to guarantee the invariance condition (\ref{eq:invariant}) and various forms of the acceptance probability in each algorithm. More importantly, given any specific target density, it is not at all clear how we can, at least conceptually, judge the pros and cons of each algorithm among an array of algorithms in the Markov chain Monte Carlo literature.

The Metropolis--Hastings algorithm is very well studied, e.g. by \cite{Roberts2004a,Meyn2009}, and widely used, as disccused by \cite{Diaconis2009}. Hence there is great service in generalising this algorithm to accomodate a larger class of Markov chain Monte Carlo algorithms. Additionally, a common derivation for many algorithms not only helps to draw comparisons and connections between them, but also provides a flexible framework to synthesise multiple sampling approaches to result in more superior algorithms. This necessary generalisation of the conventional Metropolis--Hastings algorithm will now be presented before other algorithms and their strategic combinations can be given as special cases.

\subsection{The Generalised Framework}
In order to move to an understanding of multiple sampling approaches within a unified  framework, we propose to generalise the Metropolis--Hastings algorithm by decomposing the proposal density $\gamma(\xi\mid\theta)$ into 2 sub--components. The first component is a random number generator
	\begin{equation}\label{eq:proposal_V}
		\scripted{V}\sim\scripted{q}(\scripted{V}\mid\theta),
	\end{equation} where $\scripted{V}\in\mathbb{R}^{n_\scripted{V}}$ represents all the random numbers generated in each iteration and can has as many dimensions as needed. For example, in the conventional Metropolis--Hastings proposal (\ref{eq:MH_proposal}) we can define $\scripted{V}\defeq\xi-\theta\text{, such that } \scripted{q}(\scripted{V}\mid\theta) \equiv \gamma(\xi\mid\theta)$. This change of variable is quite trivial for Markov chain Monte Carlo algorithms that are special cases of conventional Metropolis--Hastings. However, it also allows us to describe other proposals such as Gibbs sampling in (\ref{eq:Gibbs}) by defining $\scripted{V}$ as comprised of the updated dimensions of $\mathbb{R}^{n_\theta}$ as follows
	\begin{equation}\label{eq:Gibbs_prop}
		\scripted{V}\defeq\xi^d\in\mathbb{R}^{n_d} \text{, which implies } \scripted{q}(\scripted{V}\mid\theta) \equiv \pi(\xi^d\mid\theta^{-d}).
	\end{equation}
	We can also explain directional (\ref{eq:directional}) or Hamiltonian sampling (\ref{eq:Hamiltonian}) by defining $\scripted{V}\defeq\mat{\scripted{r}, &\scripted{v}}\in\mathbb{R}^{n_\theta+1}$ and constructing appropriate proposal density $\scripted{q}(\scripted{r},\scripted{v}\mid\theta)$ for each algorithm.	
	
The second component is a deterministic mapping in the expanded space of $\mat{\theta,&\scripted{V}}$ denoted as
	\begin{equation}\label{eq:T}
	\begin{aligned}
		\scripted{T}:\mathbb{R}^{n_\theta}\times\mathbb{R}^{n_\scripted{V}} &\mapsto\mathbb{R}^{n_\theta}\times\mathbb{R}^{n_\scripted{V}}
\\
		\mat{\theta, & \scripted{V}} &\mapsto \mat{\xi, & \scripted{W}},
	\end{aligned}
	\end{equation} where we are only interested in using $\scripted{T}(\theta,\scripted{V})$ to calculate the proposal $\xi$ and the auxiliary vector $\scripted{W}$ will simply be discarded. Later, we will show that the proposals (\ref{eq:MH_proposal}),(\ref{eq:Gibbs}),(\ref{eq:directional}) and (\ref{eq:Hamiltonian}) can all be seen as examples of a combination of the random draw (\ref{eq:proposal_V}) followed by the mapping (\ref{eq:T}).

The Markov transition $\gamma(\xi\mid\theta)$ can only be in fact implemented from these two aforementioned sub--components since any computer operation can only be either a deterministic calculation or a pseudo random number generation. This decomposition is inspired by \cite{Green1995} in the context of model selection, where the number of sampling dimensions changes as the Markov chain proceeds. In what follows, we will apply this idea in the context of Markov chain Monte Carlo sampling and provide the necessary proofs, which also reveal surprising symmetry that exists in many current sampling methods, especially with slice sampling in section (\ref{sec:IPG}).

Let $\theta_k\sim\pi(\cdot)$ and denote the joint density of $\mat{\theta, &\scripted{V}}$ as $\pi(\theta,\scripted{V})\defeq\pi(\theta)\scripted{q}(\scripted{V}\mid\theta)$, the generalised Metropolis--Hasting algorithm to draw $\theta_{k+1}\sim\pi(\cdot)$ can be described as follows
\begin{algorithm}[h!]
\caption{Generalised Metropolis--Hasting Algorithm}
\label{alg:GMH}
\begin{enumerate}
	\item Draw a proposal $ \scripted{V}_k \sim\scripted{q}(\scripted{V}_k\mid\theta_k) $ and compute the generalised acceptance probability as 
		\begin{equation}\label{eq:GMH_acc}
		\alpha(\theta_k,\scripted{V}_k) \defeq\MN{1,\frac{\pi(\xi_k,\scripted{W}_k)}{\pi(\theta_k,\scripted{V}_k)}\vmat{J_\scripted{T}(\theta_k,\scripted{V}_k)}};\quad (\xi_k,\scripted{W}_k) \defeq\scripted{T}(\theta_k,\scripted{V}_k),
		\end{equation} where $\vmat{J_\scripted{T}(\theta_k,\scripted{V}_k)}$ is the determinant of the Jacobian of $\scripted{T}$ evaluated at $\mat{\theta_k,\scripted{V}_k}$.
	\item For some $u_k\sim\mathcal{U}[0,1]$, we set $\theta_{k+1} = \begin{cases}
		\xi_k &\text{if }u_k\le \alpha(\theta_k,\scripted{V}_k),\\
		\theta_k &\text{otherwise}. 
	\end{cases}$
\end{enumerate}
\end{algorithm}
\begin{remark}
	We often can only compute $\widetilde{\pi}(\theta) = \pi(\theta) Z_{\pi}$ but not $\pi(\theta)$, where the normalising constant $Z_{\pi}$ is unknown. However, algorithm (\ref{alg:GMH}) can be executed without change when $\pi(\theta)$ is replaced by $\widetilde{\pi}(\theta)$ since $Z_{\pi}$ will simply be cancelled when computing the acceptance probability.
\end{remark}
\begin{theorem}\label{thm:MH}
	The invariance condition (\ref{eq:invariant}) holds for algorithm (\ref{alg:GMH}), i.e. $\theta_k{\sim}\pi(\cdot){\Rightarrow}\theta_{k+1}{\sim}\pi(\cdot),$ if the mapping $\scripted{T}(\cdot)$ is continuously differentiable	and also self--inverse in the sense that 
	\begin{equation*}
		\scripted{T}^{-1}(\cdot)\equiv\scripted{T}(\cdot)\Leftrightarrow\scripted{T}\left(\scripted{T}(\theta,\scripted{V})\right) = (\theta,\scripted{V}).
	\end{equation*}
\end{theorem}
\begin{proof}
Algorithm (\ref{alg:GMH}) is indeed invariant with respect to the joint density $\pi(\theta,\scripted{V})$ in the expanded space $\mathbb{R}^{n_\theta}\times\mathbb{R}^{n_\scripted{V}}$ and therefore is also invariant with respect to $\pi(\theta)$. Hence, let us simplify the notation by denoting the current and the next location of this Markov chain as
	\begin{equation*}
		\scripted{X} \defeq \mat{\theta, &\scripted{V}};\quad\scripted{Y}\defeq \begin{cases}
				\mat{\xi, &\scripted{W}}\defeq\scripted{T}(\scripted{X}) &\text{, if } u\le\alpha(\theta,\scripted{V}),\\
				\scripted{X} &\text{, otherwise}.
			\end{cases}
	\end{equation*}
	The inverse image of $\scripted{Y}$ is also denote as $\scripted{Z}\defeq\scripted{T}^{-1}(\scripted{Y})$.	 Given $\scripted{X}\sim\pi(\cdot)$, we can derive the distribution of $\scripted{Y}$ as follows:
	\begin{equation}\label{eq:Y_dist}
		\p{}{\dif\scripted{Y}} = \alpha(\scripted{Z})\p{}{\scripted{X}\in\dif\scripted{Z}} + \left(1-\alpha(\scripted{Y})\right)\p{}{\scripted{X}\in\dif\scripted{Y}},
	\end{equation} where the first term is the probability of starting at $\scripted{X}\in\dif\scripted{Z}$ and accepting a move to $\scripted{Y}\in\dif\scripted{Y}$ with probability $\alpha(\scripted{Z})$, and similarly, the second term is the probability of starting at $\scripted{X}\in\dif\scripted{Y}$ and rejecting a move with probability $\left(1-\alpha(\scripted{Y})\right)$. We can substitute (\ref{eq:GMH_acc}) into (\ref{eq:Y_dist}) to result in
	\begin{equation}\label{eq:theorem_1}
		\p{}{\dif\scripted{Y}} = \MN{1,\frac{\pi(\scripted{Y})}{\pi(\scripted{Z})}\vmat{J_\scripted{T}(\scripted{Z})}}\pi(\scripted{Z})\dif\scripted{Z} + \left( 1-\MN{1,\frac{\pi(\scripted{Z})}{\pi(\scripted{Y})}\vmat{J_\scripted{T}(\scripted{Y})}} \right)\pi(\scripted{Y})\dif\scripted{Y}.
	\end{equation} 
	
	Due to the mapping $\scripted{T}(\cdot)$ being self--inverse, we can apply the \emph{inverse function theorem}, e.g. see \citep{DeOliveira2013}, and find that $\vmat{J_\scripted{T}(\scripted{Z})} = \vmat{J_{\scripted{T}^{-1}}(\scripted{Y})}^{-1} = \vmat{J_\scripted{T}(\scripted{Y})}^{-1}$. This means in the case when $\frac{\pi(\scripted{Y})}{\pi(\scripted{Z})}\vmat{J_\scripted{T}(\scripted{Z})} \ge 1$ we will also have $\frac{\pi(\scripted{Z})}{\pi(\scripted{Y})}\vmat{J_\scripted{T}(\scripted{Y})} \le 1,$ so that (\ref{eq:theorem_1}) could also be simplified to
	\begin{equation}
	\label{eq:inverse_function_theorem}	
		\p{}{\dif\scripted{Y}} = \pi(\scripted{Z})\dif\scripted{Z} + \pi(\scripted{Y})\dif\scripted{Y} - \frac{\pi(\scripted{Z})}{\pi(\scripted{Y})}\vmat{J_\scripted{T}(\scripted{Y})}\pi(\scripted{Y})\dif\scripted{Y}.
	\end{equation}
	Further since $\frac{\pi(\scripted{Z})}{\pi(\scripted{Y})}\vmat{J_\scripted{T}(\scripted{Y})}\pi(\scripted{Y})\dif\scripted{Y} = \pi(\scripted{Z})\dif\scripted{Z},
$ we can again reduce (\ref{eq:inverse_function_theorem}) to precisely $$\p{}{\dif\scripted{Y}} = \pi(\scripted{Y})\dif\scripted{Y}.$$ The same conclusion is reached if $\frac{\pi(\scripted{Y})}{\pi(\scripted{Z})}\vmat{J_\scripted{T}(\scripted{Z})} < 1 $ and therefore $\scripted{Y}\sim\pi(\cdot)$ in both cases.
\end{proof}
\begin{remark}\label{rem:degenerated_mapping} Note that in the case where the $\scripted{V}$ component is degenerated, i.e. $n_{\scripted{V}} = 0$ and $\scripted{T}$ is still a self-reverse mapping in the space of $\theta\in\mathbb{R}^{n_\theta}$, the invariance condition in theorem (\ref{thm:MH}) still hold and the same proof is applicable without modification except that $\scripted{X,Y}$ simply degenerate into $\theta,\xi$ respectively. We will need this observation when explaining the connection between Metropolis-Hastings and slice sampling in section (\ref{sec:slice_sampling}).
\end{remark}

It is rather surprising that many of the existing Markov chain Monte Carlo algorithms confide with the preceding simple proof that relies on the invariance with respect to not only $\pi(\theta)$ but also the joint density $\pi(\theta,\scripted{V})$. Many other algorithms can be seen as special cases of theorem (\ref{thm:MH}) where the proposal $\scripted{q}(\scripted{V}_k\mid\theta_k)$ and the mapping $\scripted{T}(\theta_k,\scripted{V}_k)$ take on different forms. As presented in the remaining of this paper, each of the transition kernels investigated exhibits a special symmetry such that it corresponds to a self--inverse mapping, even though the original framework given by \cite{Green1995} also allows for more general mapping to be used as explained in appendix (\ref{appx:Green}). We now present some important classes of Markov chain Monte Carlo algorithms under this new framework along with their practical motivation. While a general framework to accommodate multiple algorithms is necessarily more involved than the conventional Metropolis--Hastings, we believe the main value of deriving new formulations for known algorithms is that now we can easily point out important theoretical connections that enable many advantageous combinations between those algorithms.

\section{Variants of Generalised Metropolis--Hastings}
\subsection{Metropolis--Hastings}

The conventional Metropolis--Hastings algorithm, given by equations (\ref{eq:MH_proposal})--(\ref{eq:MH}), can be seen as a special case of theorem (\ref{thm:MH}) by defining the mapping $\scripted{T}:\mat{\theta, & \scripted{V}} \mapsto \mat{\xi, & \scripted{W}}$ as follows
\begin{equation}\label{eq:T_4_MH}
	\xi \defeq \theta + \scripted{V};\quad \scripted{W} \defeq -\scripted{V};\quad\scripted{V}\in\mathbb{R}^{n_\theta} ,
\end{equation} so that the proposal density is consequently defined as $\scripted{q}(\scripted{V}\mid\theta) = \gamma(\xi\mid\theta)$. Since the self--inverse mapping (\ref{eq:T_4_MH}) has unity Jacobian, the acceptance probability (\ref{eq:GMH_acc}) reduces to the conventional acceptance probability (\ref{eq:MH_acc}). Within this class of algorithms, the Metropolis sampler is a special case when the proposal density is symmetrical in the sense that 
\begin{equation*}
	\scripted{q}(\scripted{V}\mid\theta) = \scripted{q}(\scripted{W}\mid\xi) = \scripted{q}(-\scripted{V}\mid\xi),
\end{equation*} which by substitution into (\ref{eq:GMH_acc}) will lead to the further simplified acceptance probability
\begin{equation}\label{eq:Metropolis_acc}
	\alpha(\theta,\scripted{V}) = \MN{1,\frac{\pi(\xi)}{\pi(\theta)}}.\end{equation}
A classic example of this type of proposal is the Gaussian random walk constructed by 
\begin{align}
	\scripted{V}&\sim\mathcal{N}(0,s\Sigma),\label{eq:Metropolis}\\
	\Sigma  &\defeq \frac{1}{k}\sum_{i=1}^k(\theta_i - \widebar{\theta})(\theta_i - \widebar{\theta})^{\mathsf{T}};\quad\widebar{\theta}\defeq\frac{1}{k}\sum_{i=1}^k\theta_i,\label{eq:AM_cov}
\end{align}
 where $\Sigma$ is an adaptive approximation of the global covariance of $\pi(\theta)$ as studied by \cite{Haario2001,Andrieu2008}. The scaling factor $s$ is added in (\ref{eq:Metropolis}) can be tuned using a Robbins--Monro recursive adaptation scheme, as given by \cite{Andrieu2008,Robbins1951}, to achieve an optimal running average acceptance rate $ \widebar{\alpha}\defeq\frac{1}{k}\sum_{i=1}^k\alpha(\theta_i,\scripted{V}_i) \in [0.15,0.5],$ as recommended by \cite{Roberts1997,Roberts2001}. 

While these adaptation techniques can be successful in moderately high number of dimensions, the optimal acceptance rate is still much less than unity, which leads to increased sample correlation and enlarged estimation error due to the factor $\tau_{\widehat{f}} $ in equation (\ref{eq:convergence}). Besides, Metropolis sampling also implicitly assumes that the target density is unimodal and the initial choice of $\Sigma$ approximates the true covariance matrix well. Both of these assumptions are not verifiable a priori and possibly not true in higher dimensional applications. The performance of Metropolis sampling is illustrated for a scalar density using a toy model in example (\ref{ex:Metropolis}), appendix (\ref{appx:toy_ex}).

\subsection{Gibbs Sampling}
The main advantage of Gibbs sampling is that we attain unity acceptance probability, given that we can draw exact samples from the conditional densities $\pi(\theta^d\mid\theta^{-d})$ while leaving the $\theta^{-d}$ component unchanged. For Gibbs sampling, the mapping $\scripted{T}:\mat{\theta, & \scripted{V}} \mapsto \mat{\xi, & \scripted{W}}$ is specified as 
\begin{equation*}
	\xi^{-d} = \theta^{-d};\quad\xi^d = \scripted{V};\quad\scripted{W} = \theta^d;\quad\scripted{V},\scripted{W}\in\mathbb{R}^{n_d},
\end{equation*} which is trivially self--inverse and has unity Jacobian. Consequently, the proposal density can be written as in (\ref{eq:Gibbs_prop}) so that by the same derivation in (\ref{eq:Gibbs_acc}), the acceptance probability becomes unity.

To get a full update for $\theta\in\mathbb{R}^{n_\theta}$, we simply have to update the all the subcomponents $\theta^d$ either sequentially or randomly as discussed by \cite{MacEachern2000}, while knowing that no matter how $\theta$ is divided into any number of subcomponents and which subcomponent are being updated, the invariance condition (\ref{eq:invariant}) is always satisfies according to theorem (\ref{thm:MH}). 

One major drawback for Gibbs sampling and its associated variants, such as auxiliary variables methods by \cite{Besag1993,Damien1999} or component--wise slice sampling by \cite{Neal2003}, is that these samplers can only travel along the Cartesian coordinates of $\mathbb{R}^{n_\theta}$. This behaviour can produce highly correlated samples when the target density has high spatial correlation between its dimensions. One way to circumvent this situation is to estimate the global covariance matrix, as with Metropolis sampling, and then perform an affine transformation of the coordinate system to effectively move along the eigenvectors of $\Sigma$ instead. One example for this approach in the case of slice sampling is given by \cite{Tibbits2013}. 

Gibbs--based sampling approaches therefore suffer from the same drawbacks as with Metropolis sampling, with regard to estimating the covariance structure of the target $\pi(\theta)$, but to a lesser degree thanks to the unity acceptance probability. These Gibbs--based variants are now presented.

\subsection{Generalised Metropolis--Hastings within Gibbs Sampling}
\label{sec:MH_Gibbs}
For many applications, we can draw exact samples from some but not all of the conditional densities in the subspaces of $\mathbb{R}^{n_\theta}$. Therefore, Gibbs sampling can be a more practical approach if we can relax the exact conditional sampling requirement so it is permissible that $\scripted{q}(\scripted{V}\mid\theta)\neq \pi(\xi^d\mid\theta^{-d})$ in any subspace where direct sampling from the conditional density is not feasible. 

It is probably well known that we can embed a conventional Metropolis--Hastings kernel inside a Gibbs sampling scheme. With the new framework, we can show in a fairly straightforward manner that performing any variant of the generalised Metropolis--Hastings within Gibbs sampling is indeed a special case of the generalised Metropolis--Hastings algorithm itself, with a different choice of mapping and proposal density. First, we redefine the mapping $\scripted{T}(\cdot)$ as follows
\begin{equation}\label{eq:MH_Gibbs}
	\xi^{-d} = \theta^{-d};\quad(\xi^d,\scripted{W}) \defeq \scripted{T}^d(\theta^d,\scripted{V}),
\end{equation} which leads to $\vmat{J_\scripted{T}(\theta,\scripted{V})} = \vmat{J_{\scripted{T}^d}(\theta^d,\scripted{V})},$ where any self--reverse mapping $\scripted{T}^d$ would suffice, e.g. when $\scripted{V},\scripted{W}\in\mathbb{R}^{n_d}$ we can let $\scripted{T}^d$ be given by a sub--version of (\ref{eq:T_4_MH}) restricted to $\mathbb{R}^{n_d}$ with unity Jacobian as follows
\begin{equation*}
	\xi^d = \scripted{V} + \theta^d;\quad\scripted{W} = -\scripted{V}.
\end{equation*}
Furthermore, since $\pi(\theta) = \pi(\theta^d\mid  \theta^{-d})\pi(\theta^{-d})$ and $\pi(\xi^{-d}) = \pi(\theta^{-d})$ by design, the acceptance probability now becomes identical to that of a general Metropolis--Hastings algorithm with the target being the conditional density $\pi(\theta^d\mid\theta^{-d})$ instead of the full density $\pi(\theta)$, i.e.
\begin{equation}\label{eq:MH_Gibbs_acc}
	\alpha(\theta,\scripted{V}) = \MN{1,\frac{\pi(\xi^d\mid\xi^{-d})\scripted{q}(\scripted{W}\mid\xi)}{\pi(\theta^d\mid\theta^{-d})\scripted{q}(\scripted{V}\mid\theta)}\vmat{J_{\scripted{T}^d}(\theta^d,\scripted{V})}}.
\end{equation} This means when any sub-partition of $\theta$ is fixed in a given Gibbs-liked iteration, we can apply just about any variants of Metropolis-Hastings to the conditional density given in the remaining subspace. One special advantage here is that the sampling kernel in the remaining subspace now can also be guided by the fixed components since $\scripted{q}(.)$ is conditioned upon the entire $\theta\in\mathbb{R}^{n_\theta}.$

As we will see in section (\ref{sec:directional}), being able to fit any Markov chain Monte Carlo algorithm in the framework provided by theorem (\ref{thm:MH}) into a Gibbs sampling scheme allows us to strategically combine algorithms that were previously thought to be unrelated. We illustrate the advantage of Gibbs sampling and its immediate extension to (\ref{eq:MH_Gibbs})--(\ref{eq:MH_Gibbs_acc}) in example (\ref{ex:Gibbs}), appendix (\ref{appx:toy_ex}).

\subsection{Auxiliary Variables Method}

The auxiliary variables method is a generalisation of the Swendsen--Wang algorithm for the Ising model in statistical mechanics given by \cite{Edwards1988}, which is subsequently studied by \cite{Besag1993,Higdon1998,Damien1999} and also, from a different viewpoint, by \cite{Neal1997}. In the context of high dimensional Markov chain Monte Carlo sampling, we can interpret the auxiliary variables method as a way to seperate out the irregular geometry induced by the data in the target density by introducing auxiliary random variables to transform these irregular components into conditionally uniform densities. 

Specifically, when the target density can be factored into a product of $N$ factors denoted as $\pi(\theta) = \prod_{n = 1}^N \scripted{l}_n(\theta)$, then we can introduce one independent auxiliary random variables $h_n$ for each factor $\scripted{l}_n(\theta)$ with the following conditional density
\begin{equation}\label{eq:aux_Gibbs_h}
  h_n\mid\theta\sim\mathcal{U}[0,\scripted{l}_n(\theta)], n=1,2,\dots N,
\end{equation} so that the joint density of $\mat{\theta,& \mathbb{h}\defeq\{h_n\}_{n=1}^N}$ becomes
\begin{equation}\label{eq:aux_Gibbs_joint}
  \pi(\theta,\mathbb{h})\defeq\pi(\theta)\p{}{\mathbb{h}\mid\theta} = \mathbb{1}_\Upsilon(\theta,\mathbb{h});\quad\Upsilon\defeq\{\mat{\theta,&\mathbb{h}} : 0 < h_n < \scripted{l}_n(\theta),\forall n\}.
\end{equation} 

Since $\pi(\theta)$ is a marginal density of $\pi(\theta,\mathbb{h})$, which is a uniform density as seen in (\ref{eq:aux_Gibbs_joint}), we can perform Gibbs sampling on $\pi(\theta,\mathbb{h})$ to retrieve $\theta_k\sim\pi(\cdot)$. To employ Gibbs sampling on $\pi(\theta,\mathbb{h})$, it is trivial to sample from the conditional density $\pi(\mathbb{h}\mid\theta)\defeq\prod_{n=1}^N\p{}{h_n\mid\theta}$ using (\ref{eq:aux_Gibbs_h}), while depending on the complexity of $\pi(\theta)$, we can draw exact samples from some if not all of the conditional densities for different subspaces of $\theta\in\mathbb{R}^{n_\theta}$ denoted as
\begin{equation*}
	\pi(\theta^d\mid\theta^{-d},\mathbb{h}) \propto \mathbb{1}_{\scripted{S}^d}(\theta^d);\quad\scripted{S}^d\defeq\{\theta^d : h_n < \scripted{l}_n(\theta), n=1,2,\dots N\}.
\end{equation*} 
Hence when the exact boundary of $\scripted{S}^d$ is tractable and facilitates direct sampling from, auxiliary variables method is a great way to apply Gibbs sampling to large dimensional problems, as illustrated in example (\ref{ex:aux_Gibbs}), appendix (\ref{appx:toy_ex}).

\subsection{Slice Sampling}\label{sec:slice_sampling}
Introducing too many auxiliary variables into an even trivially simple model, as seen in example (\ref{ex:aux_Gibbs}), appendix (\ref{appx:toy_ex}), can leads to excessive spatial correlation between the original parameters $\theta$ and the auxiliary dimensions $\mathbb{h}$. Therefore, \cite{Neal2003} argued that one auxiliary variable is often enough and leads to less correlated Markov samples. This so called slice sampling approach is then simply understood as a special case of auxiliary variable methods where the number of auxiliary factors is $N=1$, which leads to the following Gibbs sampling scheme
\begin{align}
	\scripted{h}\mid\theta &\sim \mathcal{U}[0,\pi(\theta)],\nonumber \\
	\theta\mid\scripted{h} &\sim \mathbb{1}_\scripted{S}(\theta);\quad\scripted{S}\defeq\{\theta : \pi(\theta) > \scripted{h}\}.\label{eq:slice}
\end{align}

While slice sampling essentially boils down to uniform sampling under the graph of $\pi(\theta)$ in $\mathbb{R}^{n_\theta}$, we can find some surprising equivalences between slice sampling and conventional Metropolis--Hastings. Let us first propose to solve the sampling task (\ref{eq:slice}) using a Metropolis sampling step in accordance with the Metropolis--Hastings within Gibbs approach as explained in section (\ref{sec:MH_Gibbs}). Specifically, we can propose $\xi = \theta + \scripted{V}$, using some symmetrical proposal density $\scripted{V}\sim\scripted{q}(\scripted{V}\mid\theta) = \scripted{q}(\scripted{W}\mid\xi)$, and compute the acceptance probability with respect to the target density (\ref{eq:slice}) as follows
\begin{equation}\label{eq:MH_Slice}
	\alpha(\theta,\scripted{V}) = \MN{1,\frac{\mathbb{1}_\scripted{S}(\xi) }{\mathbb{1}_\scripted{S}(\theta)}}= \mathbb{1}_\scripted{S}(\xi) = \begin{cases}
		1 &\text{ if }\pi(\xi) > \scripted{h},\\
		0 &\text{ otherwise.}
	\end{cases}
\end{equation}

As observed by \cite{Higdon1998}, the acceptance probability (\ref{eq:MH_Slice}) indeed corresponds to the acceptance probability (\ref{eq:Metropolis_acc}) for an ordinary Metropolis sampler with target density $\pi(\theta)$; since when we define $\scripted{h} \defeq u\pi(\theta) \sim\mathcal{U}[0,\pi(\theta)]$ for some $u\sim\mathcal{U}[0,1]$,  then (\ref{eq:MH_Slice}) is also equivalent to 
\begin{equation}\label{eq:M_Slice_acc}
	\theta_{k+1} = \begin{cases}
		\xi_k &\text{ if } u_k <{\pi(\xi_k)}/{\pi(\theta_k)},\\
		\theta_k &\text{ otherwise.}
	\end{cases}
\end{equation}
In other words, since the acceptance rate given in (\ref{eq:Metropolis_acc}) leads to the same transition rule as in (\ref{eq:M_Slice_acc}), we have shown that the Metropolis sampler is actually sampling uniformly under the graph of $\pi(\theta)$ when $\scripted{h}$ simply is a linear scaling of $u$. So in a way, we all have been doing slice sampling since the beginning of Markov Chain Monte Carlo literature without noticing the fact.

When the method was generalised by \cite{Hastings1970}, the preceding equivalence still surprisingly holds by considering an extended slice sampling scheme that draw uniform samples under the graph of the expanded target $\pi(\theta,\scripted{V})\defeq\pi(\theta)\scripted{q}(\scripted{V}\mid\theta)$ as follows  
\begin{align}
	\scripted{V}\mid\theta &\sim\scripted{q}(\scripted{V}\mid\theta)\nonumber,\\
	\scripted{h}\mid\theta,\scripted{V}&\sim\mathcal{U}[0,\pi(\theta,\scripted{V})]\nonumber,\\
	\theta,\scripted{V}\mid\scripted{h}&\sim\mathbb{1}_\scripted{S}(\theta,\scripted{V});\quad\scripted{S}\defeq\{\mat{\theta,&\scripted{V}} : \pi(\theta,\scripted{V})\ge\scripted{h}\}\label{eq:ext_slice},
\end{align} which admits $\pi(\theta,\scripted{V})$ and therefore also $\pi(\theta)$ as its invariance densities.

As mentioned during remark (\ref{rem:degenerated_mapping}), we can also solve the sampling step (\ref{eq:ext_slice}) using the mapping $\mat{\xi,&\scripted{W}}\defeq\scripted{T}(\theta,\scripted{V})$ alone, and if the Jacobian $\vmat{J_\scripted{T}}\equiv 1$, then the acceptance probability of $\mat{\xi,&\scripted{W}}$ with respect to the uniform target density (\ref{eq:ext_slice}) is computed as follows 
\begin{equation*}
	\alpha(\theta,\scripted{V}) = \mathbb{1}_\scripted{S}(\xi,\scripted{W}) = \begin{cases}
		1 & \text{if } \pi(\xi,\scripted{W})\ge\scripted{h},\\
		0 & \text{otherwise}.
	\end{cases}
\end{equation*}

Again, if we consider a similar linear scaling as before such that $\scripted{h}\defeq u\pi(\theta,\scripted{V})$ for some $u\sim\mathcal{U}[0,1]$, then the condition that ${\scripted{h}\le\pi(\xi,\scripted{W})}$ is also equivalent to $u\le\MN{1,\frac{\pi(\xi,\scripted{W})}{\pi(\theta,\scripted{V})}},$
which is identical to the conventional Metropolis--Hastings acceptance test. 

At this point, we can make one important observation regarding all variations of Metropolis--Hastings algorithm as follows. While the conventional presentation of the Metropolis--Hastings algorithm (\ref{alg:GMH}) may lead us to often thinking that $u$, or equivalently $\scripted{h}$, is generated only after the proposal $\xi$ becomes available, this ordering is indeed immaterial and could be reversed as below

\begin{align}
	\scripted{h}\mid\theta,\scripted{V}&\sim\mathcal{U}[0,\pi(\theta,\scripted{V})]\nonumber ,\\
	\scripted{V}\mid\theta,\scripted{h}&\sim\mathbb{1}_{\scripted{S}^\scripted{V}}(\scripted{V});\quad\scripted{S}^\scripted{V}\defeq\{\scripted{V} : \pi(\theta,\scripted{V})\ge\scripted{h}\}\label{eq:refresh_V} ,\\
	\theta,\scripted{V}\mid\scripted{h}&\sim\mathbb{1}_\scripted{S}(\theta,\scripted{V});\quad\scripted{S}\defeq\{\mat{\theta,&\scripted{V}} : \pi(\theta,\scripted{V})\ge\scripted{h}\}\nonumber,
\end{align}

As with this new slice sampling scheme above, we can actually condition the proposal density of $\scripted{V}$, and accordingly $\xi$, upon the specific value of $\scripted{h}$ in each iteration. For example, we should propose $\xi$ with larger distance from $\theta$ when $\scripted{h}$ is smaller and vice versa. This feature is indeed fully exploited by \cite{Neal2003} in his slice sampling scheme such that the proposal density of $\scripted{V}$ is locally adapted to both values of $\theta$ and $\scripted{h}$ as will be shown in the next section. Similarly, we will also exploit the slice sampling scheme (\ref{eq:refresh_V}) to later design the Hamiltonian slice sampler given in section (\ref{sec:HSS}).

The main contribution of \cite{Neal2003} is therefore  not only to motivate the single auxiliary variable approach, but also to design efficient methods for uniform sampling from the set $\scripted{S}$ as required in (\ref{eq:slice}) or (\ref{eq:ext_slice}, \ref{eq:refresh_V}), when an analytical solution for the boundary of $\scripted{S}$ is often not available. We illustrate the comparative performance of the single auxiliary variable approach motivated by \cite{Neal2003} in example (\ref{ex:slice}), appendix (\ref{appx:toy_ex}). These methods for uniform sampling from the set $\scripted{S}$, which are the main advantage edge in slice sampling, are now presented.

\subsection{Recursive Proposal Generation in Slice Sampling}\label{sec:IPG}
In conventional Metropolis sampling, rejection is both a necessity and a curse. The proof of theorem (\ref{thm:MH}) shows that rejections play a critical role in maintaining the invariance condition (\ref{eq:invariant}). On the other hand, rejection is also a source of sample correlation which increase the factor $ \tau_{\widehat{f}} $ given in (\ref{eq:convergence}). In the context of slice sampling algorithms, a \emph{recursive proposal generator} is an advantageous way to construct an abstract proposal density $\scripted{q}^\star(\scripted{V}^\star\mid\theta)$ such that the acceptance probability is unity regardless of the set $\scripted{S}$, i.e. we can guarantee $\theta_{k+1}\neq\theta_k$ with probability one, which helps reduce sample correlation in each slice sampling iteration. 

Specifically, for a uniform density on any bounded set $\scripted{S}\subset\mathbb{R}^{n_\theta}$ such as defined in (\ref{eq:slice}), \cite{Neal2003} provided a framework for continuing to generate subsequent proposals if the current proposal is rejected. In this framework, we can indeed generate infinitely many sequential proposals. This radical extension of the Metropolis--Hastings framework can be described in algorithm (\ref{alg:IPG}), whose condition for validity is accordingly given in theorem (\ref{thm:IPG}). 
\begin{algorithm}[h!]
\caption{Neal's Recursive Proposal Generation Algorithm}
\label{alg:IPG} For any bounded set $\scripted{S}\subset\mathbb{R}^{n_\theta}$ and assume that $\theta_k\sim\mathbb{1}_\scripted{S}(\cdot)$, we produce $\theta_{k+1}\sim\mathbb{1}_\scripted{S}(\cdot)$ (starting with $ n=1 $) as follows:
\begin{enumerate}
	\item Draw $ \scripted{V}_n \sim\scripted{q}_n(\scripted{V}_n\mid\theta_k) $ and compute $ (\xi_n,\scripted{W}_n) \defeq \scripted{T}(\theta_k,\scripted{V}_n) $ \label{st:innovation}
	\item If $ \xi_n\in\scripted{S} $, mark the successful proposal $\xi^\star_k = \xi_n $. Otherwise, repeat step (\ref{st:innovation}) with $ {n\leftarrow n+1} $
	\item Accept $\theta_{k+1} \equiv\xi^\star_k$ once $\xi^\star_k$ is generated
\end{enumerate}
\end{algorithm} 

Note that in the context of a single iteration, let us omit the subscript $k$ to improve brevity in the ensuing notations. Similarly with $\xi^\star$, we define $\scripted{V}^\star$ as the first accepted proposal in the sequence $\{\scripted{V}_n\}$ and $\scripted{W}^\star$ as the image of $\scripted{V}^\star$ through the mapping $\scripted{T}:\mat{\theta,&\scripted{V}^\star}\mapsto\mat{\xi^\star,&\scripted{W}^\star}$. We can conceptually think of algorithm (\ref{alg:IPG}) as a way to generate $\scripted{V}^\star$ from a special density denoted as 
\begin{equation*}
	\scripted{V}^\star\sim\scripted{q}^\star(\scripted{V}^\star\mid\theta) \text{ such that } \xi^\star\in\scripted{S} \text{ with probability one.}
\end{equation*}

To state the sufficient conditions for algorithm (\ref{alg:IPG}), we denote the probability of encountering a rejection in stage $n$ with the current sample being $ \theta $ as  $$ \Delta_n \defeq 1-\E{\scripted{V}_n\mid\theta}{\mathbb{1}_\scripted{S}(\xi_n)},\quad\scripted{V}_n\sim\scripted{q}_n(\scripted{V}_n\mid\theta) .$$ 
We similarly define the probability of rejection in stage $n$ with the current sample being $ \xi^\star $ as  $$ \widetilde{\Delta}_n \defeq 1-\E{\widetilde{\scripted{W}}_n\mid\xi^\star}{\mathbb{1}_\scripted{S}(\widetilde{\theta}_n)},\quad\widetilde{\scripted{W}}_n\sim\scripted{q}_n(\widetilde{\scripted{W}}_n\mid\xi^\star),$$ where $[\widetilde{\theta}_n,\widetilde{\scripted{V}}_n] \defeq \scripted{T}(\xi^\star,\widetilde{\scripted{W}}_n)$, i.e. the sequence $ \widetilde{\scripted{\theta}}_n$ or $\widetilde{\scripted{W}}_n$ represent the proposals generated with the same sampling scheme while the current sample being $\xi^\star$ instead of $\theta$.
\begin{theorem}\label{thm:IPG}
	The abstract proposal density $\scripted{q}^\star(\scripted{V}^\star\mid\theta)$ is symmetrical in the sense that 
\begin{equation}\label{eq:IPG_symmetry}
	\scripted{q}^\star(\scripted{V}^\star\mid\theta) = \scripted{q}^\star(\scripted{W}^\star\mid\xi^\star)
\end{equation}
if, given the notation $\Delta_0 = \widetilde{\Delta}_0 \defeq 1$,  the sequence of proposal densities $ \scripted{q}_n(\scripted{V}_n\mid\theta) $ satisfies the following extended Metropolis condition
\begin{equation}\label{eq:extended_symmetry}
	\scripted{q}_n(\scripted{V}^\star\mid\theta)\prod_{i=0}^{n-1}\Delta_i
=\scripted{q}_n(\scripted{W}^\star\mid\xi^\star)\prod_{i=0}^{n-1}\widetilde{\Delta}_i;\quad\forall n \in\mathbb{N}\setminus\{0\}.
\end{equation}
\end{theorem}
\begin{proof}
We can show that (\ref{eq:extended_symmetry}) indeed leads to (\ref{eq:IPG_symmetry}) and therefore $\scripted{\xi}^\star$ has an unity acceptance probability according to corollary (\ref{crl:IPG}) to be shown later. Now, we proceed by observing that the density of $ \scripted{V}^\star $ can be identified as the sum probability of the mutually exclusive events of having $ \scripted{V}^\star $ as the first accepted proposal in stage $n^{th} = 1^{st},2^{nd},\dots \infty,$ as follows
\begin{equation}\label{eq:IPG_innv_density}
\scripted{q}^\star(\scripted{V}^\star\mid\theta) = \sum\limits_{n=1}^{\infty}\left(\scripted{q}_n(\scripted{V}^\star\mid\theta)\prod\limits_{i=0}^{n-1}\Delta_i\right).
\end{equation} 
And similarly, the density of $\scripted{W}^\star$ is also written as follows 
\begin{equation}\label{eq:IPG_innv_density2}
	\scripted{q}^\star(\scripted{W}^\star\mid\xi^\star) = \sum\limits_{n=1}^{\infty}\left(\scripted{q}_n(\scripted{W}^\star\mid\xi^\star)\prod\limits_{i=0}^{n-1}\widetilde{\Delta}_i\right).
\end{equation}
Finally, substituting (\ref{eq:IPG_innv_density}--\ref{eq:IPG_innv_density2}) into (\ref{eq:IPG_symmetry}) and comparing the terms with (\ref{eq:extended_symmetry}) shows that (\ref{eq:IPG_symmetry}) holds. 
\end{proof}
\begin{corollary}\label{crl:IPG}
	Algorithm (\ref{alg:IPG}) will be \emph{invariant} with respect to the uniform density $ \mathbb{1}_\scripted{S}(\theta) $ if $\scripted{T}(\cdot)$ always has an unity Jacobian.
\end{corollary}
\begin{proof}
	Corollary (\ref{crl:IPG}) is a special extension of theorem (\ref{thm:MH}) since the acceptance probability in this case becomes unity as follows
\begin{equation*}
	\alpha(\theta,\scripted{V}^\star) = \MN{1,\frac{\mathbb{1}_\scripted{S}(\xi^\star)\scripted{q}^\star(\scripted{W}^\star\mid\xi^\star)}{\mathbb{1}_\scripted{S}(\theta)\scripted{q}^\star(\scripted{V}^\star\mid\theta)}\vmat{J_{\scripted{T}}(\theta,\scripted{V}^\star)} } = 1,
\end{equation*} which justifies why we always set $\theta_{k+1}\equiv\xi^\star_k$ as in algorithm (\ref{alg:IPG}). 
\end{proof}

According to section (\ref{sec:MH_Gibbs}), since algorithm (\ref{alg:IPG}) is a special, and albeit rather unusual, case of theorem (\ref{thm:MH}), we can embed it into the Gibbs sampling scheme (\ref{eq:slice}) to show that slice sampling is indeed a special case of generalised Metropolis--Hastings, where the original target density $\pi(\theta)$ is transformed into a uniform joint density $\mat{\scripted{h},&\theta}\sim\mathbb{1}_\Upsilon(\cdot)$ by introducing $\scripted{h}\sim\mathcal{U}[0,\pi(\theta)]$.

The levels of symmetry required in theorem (\ref{thm:IPG}) seems rather prohibitive but nonetheless can be satisfied in a fairly general manner. Essentially, condition (\ref{eq:extended_symmetry}) can be guaranteed if the probability of arriving at $\xi^\star$ starting from $\theta$ and vice versa is equal given any combination of rejected samples $\xi_n$ in between. We now present some specific cases of this approach.

First, let us consider the mapping $\scripted{T}:\mat{\theta,&\scripted{V}_n}\mapsto\mat{\xi_n,&\scripted{W}_n}$ as with Metropolis--Hastings where
\begin{equation}\label{eq:lemma1}
	\xi_n = \theta + \scripted{V}_n;\quad\scripted{W}_n = -\scripted{V}_n;\quad\scripted{V}_n\in\mathbb{R}^{n_\theta}.
\end{equation} Then we can construct $\scripted{q}_n(\scripted{V}_n\mid\theta)$ from the sequence of posterior densities of the following artificial inference problem. First, we generate a sequence of artificial random vectors 
\begin{equation*}
	\scripted{C}_i\sim\beta_i(\scripted{C}_i\mid\scripted{D}_\circ);\quad\scripted{C}_i,\scripted{D}_i\in\mathbb{R}^{n_\scripted{V}};\quad i = 1,2,\dots n,
\end{equation*} where $\scripted{D}_\circ$ represents some arbitrary choice of coordinate origin, e.g. ${\scripted{D}_\circ = 0}$ is actually a good option. Next, we sample $\scripted{V}_n$ from a sequence of posterior densities of $\scripted{D}_\circ$ given the prior $\eta(\scripted{D})$ and data $\scripted{C}_{1:n}\defeq\{\scripted{C}_i\}_{i=1}^n$ as follows
\begin{equation}\label{eq:artificial_posterior}
	\scripted{V}_n = \scripted{D}_n - \scripted{D}_\circ;\quad \scripted{D}_n \mid\scripted{C}_{1:n} \sim \frac{\eta(\scripted{D}_n) \prod_{i=1}^n\beta_i(\scripted{C}_i\mid\scripted{D}_n)}{Z_n},
\end{equation} where $Z_n$ is the corresponding finite normalising constant. Practically, it is best to put ${\scripted{D}_\circ = 0}$ so that $ \scripted{V}_n \equiv \scripted{D}_n$, but we will maintain the general context for notational clarity. Similarly with $\scripted{V}^\star$, we also denote $\scripted{D}^\star$ as the first draw that results an accepted sample $\xi^\star$. 

\begin{theorem}\label{thm:ALP}
	The extended Metropolis condition (\ref{eq:extended_symmetry}) will be satisfied by generating $ \scripted{V}_n $ from the artificial posterior density (\ref{eq:artificial_posterior}) while the mapping $\scripted{T}:\mat{\theta,&\scripted{V}_n}\mapsto\mat{\xi_n,&\scripted{W}_n}$ is given by (\ref{eq:lemma1}).
\end{theorem}
\begin{proof}
The procedure in (\ref{eq:artificial_posterior}) explains how we simulate samples from $\scripted{V}^\star\sim\scripted{q}^\star(\scripted{V}^\star\mid\theta)$. However, the abstract proposal $\scripted{W}^\star\sim\scripted{q}^\star(\scripted{W}^\star\mid\xi^\star)$ in the reverse direction must be understood with care. First, the artificial vectors $\widetilde{\scripted{C}}_i$, must now be conditioned upon $\scripted{D}^\star$ instead of $\scripted{D}_\circ$, since their roles are now in reverse, which means $\widetilde{\scripted{W}}_n = \widetilde{\scripted{D}}_n - \scripted{D}^\star.$ This arrangement actually satisfy the mapping (\ref{eq:lemma1}) since letting $\widetilde{\scripted{D}}_n = \scripted{D}_\circ$ will lead to 
\begin{align*}
	\widetilde{\theta}_n 
	&= \xi^\star + \widetilde{\scripted{W}}_n \\
	&= (\theta + \scripted{V}^\star ) + (\widetilde{\scripted{D}}_n - \scripted{D}^\star) \\
	&= \theta + (\scripted{D}^\star - \scripted{D}_\circ) + (\scripted{D}_\circ - \scripted{D}^\star) = \theta
\end{align*}

Now we observe that $\forall n\in\mathbb{N}\setminus\{0\}$ and denoting $\scripted{C}_{1,n} \defeq [\scripted{C}_1,\scripted{C}_2,\dots\scripted{C}_n]$ and similarly for $\xi_{1,n-1},\scripted{\widetilde{C}}_{1,n},\widetilde{\theta}_{1,n-1}$, equation (\ref{eq:extended_symmetry}) is equivalent to 
\begin{equation*}
	\int\p{}{\scripted{C}_{1,n},\xi_{1,n-1},\xi_n = \xi^\star\mid\theta}\dif\scripted{C}_{1,n}\dif\xi_{1,n-1} = \int\p{}{\scripted{\widetilde{C}}_{1,n},\widetilde{\theta}_{1,n-1},\widetilde{\theta}_n = \theta\mid\xi^\star}\dif\scripted{\widetilde{C}}_{1,n}\dif\widetilde{\theta}_{1,n-1},	
\end{equation*} which can be guaranteed by an even stricter condition as follows
\begin{equation}\label{eq:lemma}
	\p{}{\scripted{C}_{1,n},\xi_{1,n-1},\xi_n = \xi^\star\mid\theta} = \p{}{\scripted{\widetilde{C}}_{1,n}=\scripted{C}_{1,n},\widetilde{\theta}_{1,n-1}=\xi_{1,n-1},\widetilde{\theta}_n = \theta\mid\xi^\star},
\end{equation} which means that the probability of transitioning from $\theta$ to $\xi^\star$ and vice versa are the same if the ``stepping stones", e.g. $\scripted{C}_{1,n},\xi_{1,n-1}$, are the same in both directions. 

In equation (\ref{eq:lemma}), we can put $\scripted{D}_{1,n-1},\widetilde{\scripted{D}}_{1,n-1}$ in places of $\xi_{1,n-1},\widetilde{\theta}_{1,n-1}$ respectively and the meaning of that equation would not change since these symbols would still depict the very same probabilistic event. This is because even though the starting positions in opposite directions are difference, being $\theta,\xi^\star$ respectively, the intermediately generated proposals will be identical, i.e. $$\widetilde{\theta}_i = \xi_i \Leftrightarrow	\widetilde{\scripted{D}}_i = \scripted{D}_i, \qquad\forall i = 1,\dots (n-1),$$ because
\begin{align*}
	\widetilde{\theta}_i 
	&= \xi^\star + \widetilde{\scripted{W}}_i = (\theta + \scripted{V}^\star) + (\widetilde{\scripted{D}}_i - \scripted{D}^\star)\\ 
	&= \theta + (\scripted{D}^\star - \scripted{D}_\circ) + \scripted{D}_i - \scripted{D}^\star\\
	&= \theta + (\scripted{D}_i - \scripted{D}_\circ) = \theta + \scripted{V}_i = \xi_i.
\end{align*}
That means because of the formulation of $\scripted{T}:\mat{\theta,&\scripted{V}_i}\mapsto\mat{\xi_i,&\scripted{W}_i}$ as in (\ref{eq:lemma1}), we then can rely on a statistical equivalency between the density of $\xi_i$ and $\scripted{D}_i$ in the sense that $\p{}{\xi_i\mid\theta} = \p{}{\scripted{V}_i\mid\theta} = \p{}{\scripted{D}_i\mid\scripted{D}_\circ}$, or analogously that $\p{}{\widetilde{\theta}_i\mid\xi^\star} = \p{}{\widetilde{\scripted{W}}_i\mid\xi^\star} = \p{}{\widetilde{\scripted{D}}_i\mid\scripted{D}^\star}$.
	
	The given understanding can now help us to analyse both sides of equation (\ref{eq:lemma}) as follows
		\begin{align}
			\p{}{\scripted{C}_{1,n},\xi_{1,n-1},\xi_n = \xi^\star\mid\theta} & = \prod_{i=1}^n\p{}{\scripted{C}_i\mid\theta}\p{}{\xi_i\mid\scripted{C}_{1:i},\theta},\nonumber \\
			& = \prod_{i=1}^n\p{}{\scripted{C}_i\mid\theta}\p{}{\scripted{D}_i\mid\scripted{C}_{1:i}},\nonumber \\
			& = \prod_{i=1}^n \beta_i(\scripted{C}_i\mid\scripted{D}_\circ) \frac{\eta(\scripted{D}_i)\prod_{j=1}^i\beta_j(\scripted{C}_j\mid\scripted{D}_i)}{Z_i}.\label{eq:LHS}
		\end{align}
Similarly, the right hand side of (\ref{eq:lemma}) is also expanded as follows
		\begin{align}
			\p{}{\scripted{\widetilde{C}}_{1,n},\widetilde{\theta}_{1,n-1},\widetilde{\theta}_n = \theta\mid\xi^\star}
			& = \prod_{i=1}^n\p{}{\scripted{\widetilde{C}}_i\mid\xi^\star}\p{}{\widetilde{\theta}_i\mid\scripted{\widetilde{C}}_{1:i},\xi^\star},\nonumber \\
			& = \prod_{i=1}^n \beta_i(\scripted{\widetilde{C}}_i\mid\scripted{D}^\star)\frac{\eta(\widetilde{\scripted{D}}_i)\prod_{j=1}^i\beta_j(\scripted{\widetilde{C}}_j\mid\widetilde{\scripted{D}}_i)}{Z_i}.\label{eq:RHS}
		\end{align} 
Comparing (\ref{eq:LHS}) and (\ref{eq:RHS}) term by term while noting that 
\begin{align*}
	\scripted{\widetilde{C}}_{1,n}=\scripted{C}_{1,n};\quad
	\widetilde{\theta}_{1,n-1}=\xi_{1,n-1};\quad
	\scripted{\widetilde{D}}_{1,n-1}=\scripted{D}_{1,n-1};\quad
	\scripted{D}_n = \scripted{D}^\star;\quad
	\scripted{\widetilde{D}}_n = \scripted{D}_\circ
\end{align*}
will reveal that (\ref{eq:lemma}) holds $\forall n\in\mathbb{N}\setminus\{0\}$, which means (\ref{eq:extended_symmetry}) also holds. See \citep{Neal2003} for a detailed exposition of the case $n=2$ and $\eta(\cdot)$ is a constant.
\end{proof}

In the simplest instance as mentioned by \cite{Neal2003}, $ \beta_i(\scripted{C}_i\mid\scripted{D}_\circ) $ can simply be Gaussian densities with mean $\scripted{D}_\circ = 0$, i.e. $\scripted{C}_i\sim\mathcal{N}(0,s\Sigma)$ and if $\eta(\cdot)$ is constant, then the posterior density (\ref{eq:artificial_posterior}) will be a product of Gaussian functions which can be reduced to the following density
\begin{equation}\label{eq:Gaussian_Slice}
	\scripted{V}_n\equiv\scripted{D}_n\sim\mathcal{N}\left(\frac{1}{n}\sum_{i=1}^n\scripted{C}_i,\frac{1}{n}s\Sigma\right),
\end{equation} which means $\V{}{\scripted{V}_n} = \frac{2}{n}s\Sigma,\forall n\ge 1$. Hence $\V{}{\scripted{V}_n}\to 0 $ as $ n\to\infty$,
 which makes the sequence $\xi_n$ increasingly concentrated around $\theta$ as $n\to\infty$. In this case, we can verify that the abstract density $\scripted{q}^\star(\scripted{V}^\star\mid\theta)$ actually integrates to one since $\Delta_n\to 0 $ as $n\to\infty$ and therefore
\begin{equation*}
\begin{aligned}
\int\scripted{q}^\star(\scripted{V}^\star\mid\theta)\dif\scripted{V}^\star&=\sum\limits_{n=1}^{\infty}\left(\int_\scripted{A} \scripted{q}_n(\scripted{V}^\star\mid\theta)\dif\scripted{V}^\star\prod_{i=0}^{n-1}\Delta_i\right),\quad\scripted{A}\defeq\{\scripted{V}^\star : \scripted{V}^\star + \theta \in \scripted{S} \}, \\
&=\sum\limits_{n=1}^{\infty}(1-\Delta_n)\prod_{i=0}^{n-1}\Delta_i =\sum\limits_{n=0}^{\infty} \Delta_{0:n} - \Delta_{0:(n+1)}= 1-\Delta_{0:\infty} = 1,
\end{aligned}
\end{equation*} where $ \Delta_{0:\infty} \defeq \prod_{n=0}^{\infty}\Delta_{n} = 0 $. In general, it is necessary to check that the proposal density of the artificial data $ \beta_i(\scripted{C}_i\mid\scripted{D}_\circ) $ actually leads to a well-defined density for the random vector $\scripted{V}^\star$. 

In light of the fact that Metropolis sampling is indeed identical to slice sampling without recursive proposal generation as explained by (\ref{eq:MH_Slice})--(\ref{eq:M_Slice_acc}), we can argue that the popularity of Metropolis sampling in present day despite the existence of a more advantageous sampler is a paradox in Markov chain Monte Carlo literature. Having a common framework for both sampling approaches helps illuminate the connection between Metropolis and slice sampling and therefore should encourage the proliferation of the more advanced algorithm.

In comparison with the conventional Metropolis proposal (\ref{eq:Metropolis}), the recursive proposal (\ref{eq:Gaussian_Slice}) is evidently more superior since the proposal sequence (\ref{eq:Gaussian_Slice}) is similar with (\ref{eq:Metropolis}) but the covariance matrix in (\ref{eq:Gaussian_Slice}) is two times the covariance in (\ref{eq:Metropolis}) when $n=1$ and eventually shrink to zero as $n\to\infty$. Hence the recursive sampling scheme allows ambitiously distant proposals while still guarantees that $\theta_{k+1}\neq\theta_k$ with probability one. The same strategy is applicable to any subspace of $\mathbb{R}^{n_\theta}$, or indeed any directional cross section of the ``slice" $\scripted{S}$ as we will see in the next section.

\subsection{Directional Sampling}
\label{sec:directional}

Directional sampling as described by \cite{Roberts1994} or differently by \cite{Chen1996a} encompasses a very large class of sampling strategies including \emph{differential evolution} by \cite{Braak2006} or \emph{t-walk} by \cite{Christen2010}. In this section, we aim to present the general formulation of directional sampling without going through the details of each sampler in this group. More importantly, we later show how this approach can be blended harmonically with slice sampling to result in an efficient, robust and autonomous algorithm.

In each directional sampling iteration, we generate the following random vectors
\begin{equation*}
	\scripted{V} \defeq\mat{\scripted{r} & \scripted{v}};\quad \scripted{W}\defeq\mat{\scripted{s} & \scripted{w}};\quad\scripted{v},\scripted{w}\in\mathbb{R}^{n_\theta};\quad\scripted{r},\scripted{s}\in\mathbb{R},
\end{equation*}
 and for a constant $\rho = 0,$ or $-1$, we define the mapping $\scripted{T}:\mat{\theta,&\scripted{V}}\mapsto\mat{\xi,&\scripted{W}}$ as follows
\begin{equation}\label{eq:ADS}
		\xi = \theta + \scripted{r}(\scripted{v} + \rho\theta);\quad\scripted{s} = \frac{-\scripted{r}}{1+\scripted{r}\rho};\quad \scripted{w} = \scripted{v}.
\end{equation} 
Assuming that $\scripted{r}$ has a well--defined density in $\mathbb{R}$ so that $(1+\scripted{r}\rho)\neq 0$ almost surely, then the mapping $\scripted{T}$ can be shown to be self--reverse with the following Jacobian \citep{Roberts1994} $$\vmat{J_\scripted{T}(\theta,\scripted{V})} = \vmat{1+\scripted{r}\rho}^{n_\theta - 2},\quad n_\theta\ge 2.$$

The proposal density for $\scripted{V}$ can be constructed as $\scripted{q}(\scripted{V}\mid\theta)\defeq\scripted{q}_\scripted{v}(\scripted{v})\scripted{q}_\scripted{r}(\scripted{r}\mid\scripted{v},\theta)$, where $\scripted{v} = \scripted{w}$ is independent from both $\theta$ or $\xi$, e.g. as in \citep{Roberts1994,Chen1996a}. This leads to $\scripted{q}_\scripted{v}(\scripted{v})\equiv\scripted{q}_\scripted{v}(\scripted{w})$ and therefore the acceptance probability in this case becomes
\begin{equation*}
	\alpha(\theta,\scripted{V}) = \MN{1,\frac{\pi(\xi)\scripted{q}_\scripted{r}(\scripted{s}\mid\scripted{v},\xi)}{\pi(\theta)\scripted{q}_\scripted{r}(\scripted{r}\mid\scripted{v},\theta)}\vmat{1+\scripted{r}\rho}^{n_\theta - 2}}.
\end{equation*}

When $\scripted{v}$ is chosen as the difference between two other randomly selected parallel Markov chains targeting the same density $\pi(\theta)$, as described by \cite[section 2.7]{Gilks1994}, the performance of directional sampling is invariant with respect to any affine transformation of the coordinates of $\mathbb{R}^{n_\theta}$ \citep[section 2.5]{Gilks1994}. This means there is no need for adaptively tuning the sampling directions by e.g. approximating the global covariance matrix of $\pi(\cdot)$, which is not robust with respect to multimodal or high dimensional densities. We now seek to combine this advantage with the recursive sampling scheme in slice sampling.
 
\subsection{Directional Slice Sampling}
A combination between directional sampling and slice sampling will be significantly more robust than either of the two individual algorithms since directional sampling is an affine invariant algorithm, while difficulties in estimating the covariance of $\pi(\cdot)$ is the common drawback of many Markov chain Monte Carlo algorithms, such as adaptive Metropolis sampling \citep{Haario2001}, elliptical or factorial slice sampling \citep{Nishihara2014a,Tibbits2013}. Meanwhile, slice sampling helps scaling the proposal density automatically to guarantee unity acceptance probability. Hence there is also no need for optimal scaling strategies such as those given by \cite{Roberts2001} for conventional Metropolis--Hastings sampling.

We now make the necessary choices in the general formulation of directional sampling to arrive at a special case that can be combined with slice sampling. Specifically, we set $\rho=0$ so that the Jacobian is $\vmat{1+\scripted{r}\times 0}^{n_\theta - 2} = 1,$ and the mapping $\scripted{T}:\mat{\theta,&\scripted{r}}\mapsto\mat{\xi,&\scripted{s}}$ is reduced to
\begin{equation}\label{eq:PADS}
		\xi = \theta + \scripted{r}\scripted{v};\quad\scripted{s} = -\scripted{r}, 
\end{equation} where the random direction $\scripted{v}\equiv\scripted{w}$ is notationally considered as a constant since it remains unchanged in each sampling iteration. Again, $\scripted{v}$ is randomly constructed from other parallel Markov chains according to \cite[section 2.7]{Gilks1994}. Now since directional sampling is a special case of theorem (\ref{thm:MH}), we can embed it into the Gibbs sampling scheme (\ref{eq:slice}) to draw samples from the uniform density $\mathbb{1}_\scripted{S}(\theta)$. This combination is called directional slice sampling.

Since all necessary conditions for recursive proposal generation in algorithm (\ref{alg:IPG}) are now available, we can construct a special proposal density $\scripted{r}^\star\sim\scripted{q}^\star_\scripted{r}(\scripted{r}^\star\mid\theta)$ that guarantees an unity acceptance probability using a recursive scheme similar to (\ref{eq:artificial_posterior}) as follows 
\begin{equation}\label{eq:artificial_posterior2}
	\scripted{r}_n\mid\scripted{C}_{1:n} \sim \frac{\prod_{i=1}^n\p{}{\scripted{C}_i\mid\scripted{r}_n}}{Z_n};\quad\scripted{C}_i\sim\beta_i(\scripted{C}_i\mid\scripted{D}_\circ\defeq 0);\quad\scripted{C}_i\in\mathbb{R}.
\end{equation}
By noticing that the mapping given in (\ref{eq:PADS}) is simply a univariate version of (\ref{eq:lemma1}) projected onto the direction $\scripted{v}$, we can show that condition (\ref{eq:extended_symmetry}) in theorem (\ref{thm:IPG}) holds for the sampling scheme (\ref{eq:artificial_posterior2}) according to the same line of reasoning as seen in theorem (\ref{thm:ALP}). 

Again, the densities $\beta_i(\scripted{C}_i\mid\scripted{D}_\circ)$ can be Gaussian functions. We can also use another alternative given by \cite{Neal2003} to result in uniform posterior densities for $\scripted{r}_n \mid\scripted{C}_{1:n}$ instead. Specifically, the random vectors $\scripted{C}_i$ are constructed as random intervals on the real line, i.e. $ \scripted{C}_i := [a_i,b_i]\subset\mathbb{R},$ where $[a_i,b_i]$ are generated by using algorithm (\ref{alg:expand_slice}), which can be shown to be a special case of (\ref{eq:artificial_posterior2}) by \cite[section V]{Tran2015b} or proven differently by \cite[section 4.3]{Neal2003}. 

In conclusion, the general directional slice sampling algorithm can be given as algorithm (\ref{alg:directional_slice}) where we always have the freedom to chose different random sampling schemes for $\scripted{v}\sim\scripted{q}_\scripted{v}(\scripted{v})$ and $\beta_i(\scripted{C}_i\mid\scripted{D}_\circ)$. This algorithm is indeed a perfect choice of Markov chain Monte Carlo kernel in the sequential Monte Carlo approach, described by \cite{Chopin2002}, since it is inherently a multi-chains sampling approach that does not require estimating the covariance matrix of its target density; it is also free from manual tuning and has unity acceptance probability.
\begin{algorithm}
	\caption{Automatic Expanding/Shrinking Procedure}
	\label{alg:expand_slice}
	First, we generate $[a_1,b_1]$ as follows
	\begin{enumerate}
		\item Setting $ a_1 = u-1 ;\;	b_1 = u $ for some $ u\sim\mathcal{U}[0,1] $
		\item While $ (\theta_{k}+a_1 \scripted{v}_k)\in\scripted{S} $, adjust $ a_1 \gets  a_1 - 1 $
		\item While $ (\theta_{k}+b_1\scripted{v}_k)\in\scripted{S}$, adjust $ b_1 \gets  b_1 + 1 $
	\end{enumerate}
	Second, we adaptively generate $[a_n,b_n]$, starting with $n=1$, by shrinking $[a_1,b_1]$ as follows
	\begin{enumerate}
		\item Sample $ \scripted{r}_n\sim\mathcal{U}[a_n,b_n] $ and compute $ \xi_n = \theta_k + \scripted{r}_n \scripted{v}_k $ \label{st:sampling_i}
		\item Set $ \theta_{k+1} = \xi_n $ if $ \xi_n\in\scripted{S}$ and terminate the routine
		\item Otherwise, repeat step (\ref{st:sampling_i}) with $ [a_{n+1},b_{n+1}] $ generated by
		\begin{enumerate}
			\item If $ \scripted{r}_{n}\ge 0 $, then set $ b_{n+1} =\scripted{r}_{n};\;	a_{n+1} = a_{n} $
			\item Otherwise, set $	b_{n+1} = b_{n};\;	a_{n+1} =\scripted{r}_{n} $
		\end{enumerate}
	\end{enumerate}
\end{algorithm} 
\begin{algorithm}
	\caption{Directional Slice Sampling}
	\label{alg:directional_slice}
	Let $\theta_k\sim\pi(\cdot) $ and starting with $ n=1 $, we produce $\theta_{k+1}\sim\pi(\cdot)$ by:
\begin{enumerate}
	\item Drawing 
	\begin{equation*}
		\begin{aligned}
			\scripted{h}_k &\sim\mathcal{U}[0,\pi(\theta_k)];&\scripted{v}_k &\sim\scripted{q}_\scripted{v}(\scripted{v}_k);\\
			\scripted{r}_n\mid\scripted{C}_{1:n} &\sim \frac{\prod_{i=1}^n\p{}{\scripted{C}_i\mid\scripted{r}_n}}{Z_n};&\scripted{C}_i &\sim\beta_i(\scripted{C}_i\mid\scripted{D}_\circ\defeq 0)
		\end{aligned}
	\end{equation*}
	\item Set $\theta_{k+1} \gets \xi_n \defeq (\theta_k + \scripted{r}_n\scripted{v}_k) $ if $ \pi(\xi_n)\ge \scripted{h}_k$
	\item Otherwise, regenerate $\scripted{r}_n$ with $ {n\leftarrow n+1} $
\end{enumerate}
\end{algorithm}
 
\subsection{Langevin and Hamiltonian Monte Carlo Sampling}
There has long been a growing interest in Markov chain Monte Carlo algorithms that exploits local gradient information to construct fast mixing proposals as seen from \cite{Roberts2002,Girolami2011,Neal2012,Hoffman2014}. Recently, automatic differentiation technology is getting really mature so that it is now practical to compute the gradient of any smooth target density to exact machine precision, within a small multiples of the computing cost for the original density, e.g. using the Stan math library by \cite{Carpenter2015}. There are also situations when the gradient of $\pi(\theta)$ is analytically tractable so that gradient--based algorithms can be even more efficient. Metropolis--adjusted Langevin and Hamiltonian Monte Carlo algorithms are two popular options in this class of algorithms.

In Metropolis--adjusted Langevin algorithm, the proposal density is constructed from a discretised solution to the Langevin diffusion equation in physics \citep{Kennedy1990} as follows
\begin{equation}\label{eq:Langevin} 
\xi = \theta+\scripted{v};\quad\scripted{v}\mid\theta \sim\mathcal{N}\left(\dfrac{\scripted{r}^2}{2}\Sigma\nabla\log\pi(\theta),\scripted{r}^2\Sigma\right) ,
\end{equation} where $\Sigma$ is often chosen e.g. as a constant diagonal matrix by \cite{Neal2012} or as the inverse observed information matrix by \cite{Girolami2011}, i.e. $\Sigma \defeq\left[-\nabla^2\log\pi(\theta)\right]^{-1}$, which equates to the covariance matrix of $\pi(\theta)$ in case the target is simply a Gaussian density. 

This algorihtm can be seen as either an instance of conventional Metropolis--Hastings or, alternatively, a truncated version of Hamiltonian Monte Carlo with a fictional random momentum $\scripted{v}\mid\theta\sim\mathcal{N}(0,\Sigma^{-1})$ \citep{Girolami2011}. When $\Sigma$ is well matched with the geometry of $\pi(\theta)$, the given choice of covariance $\Sigma^{-1}$ for the momentum vector $\scripted{v}$ has long been motivated in molecular dynamics by \cite{Bennett1975} to compensate for spatial correlations in the target density. Specifically, the proposal (\ref{eq:Langevin}) is equivalent to the following Hamiltonian proposal
\begin{equation*}
	\begin{aligned}
		\xi &= \theta + \frac{\scripted{r}^2}{2}\Sigma\nabla\log\pi(\theta)+\scripted{r}\Sigma\scripted{v},\quad\scripted{v}\mid\theta\sim\mathcal{N}(0,\Sigma^{-1}), \\
		\scripted{w} &= \scripted{v} + \frac{\scripted{r}}{2}\nabla\log\pi(\theta)+ \frac{\scripted{r}}{2}\nabla\log\pi(\xi), 
	\end{aligned}
\end{equation*} which is indeed a one step discretised solution to (\ref{eq:Hamiltonian}), see \citep[section 5.2]{Neal2012} for details. 

Therefore, we will now consider Hamiltonian sampling only while noting that Metropolis--adjusted Langevin algorithm is a special case. In Hamiltonian Monte Carlo sampling, the proposal $\xi$ is also constructed from the Hamiltonian dynamical equations (\ref{eq:Hamiltonian}), whose solution for some arbitrary interval $t\in[0,\scripted{r}]$, starting at $\theta(t=0) \defeq \theta;\;\scripted{v}(t=0)\defeq \scripted{v} ,$ is a deterministic trajectory through the space of $\mat{\theta, & \scripted{v} }\in\mathbb{R}^{n_\theta}\times\mathbb{R}^{n_\scripted{v}}$ that terminates at 
\begin{equation*}
	\theta(t=\scripted{r}) \defeq \xi;\quad\scripted{v}(t=\scripted{r}) \defeq \scripted{w}.
\end{equation*}

This trajectory naturally defines a mapping $\scripted{T}:\mat{\theta, & \scripted{v}, & \scripted{r}} \mapsto \mat{\xi, & \scripted{w}, & -\scripted{r}}$ which, according to \cite{Neal2012}, can be shown to be self--reverse with unity Jacobian, even when the mapping is computed by discretising the dynamical equations (\ref{eq:Hamiltonian}) using e.g. leapfrog method as described by \cite{Leimkuhler2005}. Therefore, we can collect the associated random vectors into $$\scripted{V}\defeq\mat{\scripted{v}, & \scripted{r}};\quad\scripted{W}\defeq\mat{\scripted{w}, & \scripted{s}};\quad\scripted{s} = -\scripted{r},$$ and design a symmetrical distribution for $ \scripted{r} $ such that 
\begin{equation}\label{eq:r_symmetry}
	\scripted{q}_\scripted{r}(\scripted{r}\mid\theta,\scripted{v}) = \scripted{q}_\scripted{r}(\scripted{s} \mid\xi,\scripted{w})= \scripted{q}_\scripted{r}(-\scripted{r}\mid\xi,\scripted{w}),
\end{equation}
so the acceptance probability for this algorithm can be computed as
\begin{equation}
\alpha(\theta,\scripted{V}) = \MN{1,\dfrac{\pi(\xi)}{\pi(\theta)}\dfrac{\scripted{q}_\scripted{v}(\scripted{w}\mid\xi)}{\scripted{q}_\scripted{v}(\scripted{v}\mid\theta)}}.
\end{equation}

When the mapping $\scripted{T}:\mat{\theta, & \scripted{v}, & \scripted{r}} \mapsto \mat{\xi, & \scripted{w}, & -\scripted{r}}$ is analytically tractable, the acceptance probability will be identically unity since the Hamiltonian quantity $ H(\theta,\scripted{v})\defeq -\log\left( \pi(\theta)\scripted{q}_\scripted{v}(\scripted{v}\mid\theta) \right)$ will be preserved. Therefore the joint density of $\mat{\theta,&\scripted{v}}$ is also invariant along this trajectory, i.e.
\begin{equation} \label{eq:H_preservation}
\pi(\theta)\scripted{q}_\scripted{v}(\scripted{v}\mid\theta) = \pi(\xi)\scripted{q}_\scripted{v}(\scripted{w}\mid\xi),\;\forall\scripted{r}\in\mathbb{R}.
\end{equation}

The property (\ref{eq:H_preservation}) will not hold when $\scripted{T}$ is approximated by a time discretisation scheme. Hence, the acceptance probability will be lowered. In discretised Hamiltonian Monte Carlo, besides tuning the covariance matrix $\Sigma^{-1}$ for the fictional momentum density $ \scripted{v}\mid\theta\sim\mathcal{N}(0,\Sigma^{-1}) $, we also need to tune the integration time $\scripted{r}$ and a discretisation step size to minimise computation cost, while maximising the expected distance between $\theta$ to $\xi$ and maintaining an average acceptance rate near the theoretical optimum of 0.65 \citep{Beskos2013a} or up to 0.85 using the ``windows of states" method given by \cite{Neal1994,Neal2012}. The necessary tuning is automated by \cite{Hoffman2014} and eventually implemented in a general purpose, free and open-source  program for Bayesian inference or optimisation named Stan \citep{Carpenter2016}.

\subsection{Elliptical Hamiltonian Slice Sampling}\label{sec:HSS}

When feasible, exact Hamiltonian dynamics solution is a fast and effective approach to construct the necessary proposal in Markovian sampling. Meanwhile, slice sampling can always help guarantee unity acceptance probability. Therefore, it is advantageous to combine these two approaches to result in algorithms with little to no tuning that can be applied in other contexts. For example, the sequential Monte Carlo approach described by \cite{Chopin2002} requires a Markov kernel with as little tuning as possible since only one or a few Markovian sampling operation is performed before changing the target density to the next density in a sequence.

While the true target $\pi(\theta)$ does not always produce an analytical solution for the Hamiltonian trajectory, we can always replace $\pi(\theta)$ in the Hamiltonian equation (\ref{eq:Hamiltonian}) with a Gaussian approximation $\psi(\theta)\equiv\mathcal{N}(\mu,\Sigma)$ such that, given the fictional momentum $\scripted{v}\mid\theta\sim\mathcal{N}(0,\Sigma^{-1})$, the approximate Hamiltonian term is defined up to some constant as 
\begin{equation}\label{eq:approx_Hamiltonian}
	\widetilde{\scripted{H}}(\theta,\scripted{v}) \defeq - \log\left(\psi(\theta)\scripted{q}_\scripted{v}(\scripted{v}\mid\theta)\right)
= \frac{1}{2}\left([\theta-\mu]^\mathsf{T}\Sigma^{-1}[\theta-\mu] + \scripted{v}^\mathsf{T}\Sigma\scripted{v}\right),
\end{equation} which results in the following closed--form Hamiltonian trajectory (see appendix \ref{appx:Hamiltonian})
\begin{equation}\label{eq:ellipse}
\begin{aligned}
	\xi &= [\theta-\mu]\cos(\scripted{r}) + \Sigma\scripted{v}\sin(\scripted{r}) + \mu ;\quad\forall\scripted{r}\in\mathbb{R},\\
		\scripted{w} &= \scripted{v}\cos(\scripted{r}) - \Sigma^{-1}[\theta-\mu]\sin(\scripted{r}).
\end{aligned}
\end{equation}

We now present some known and also novel approaches to perform slice sampling on the trajectory (\ref{eq:ellipse}). Let us first start with the initial observation by \cite{Neal1999} while studying some Gaussian process models with the target density having the popular form 
\begin{equation}\label{eq:tilde_pi}
	\pi(\theta) = \psi(\theta)\scripted{L}(\theta),
\end{equation} where $\scripted{L}(\theta)$ is a fairly flat likelihood function in comparison with the prior density $\psi(\theta)$, which therefore also captures the geometry of the posterior density $\pi(\theta)$ well. The prior density is often chosen as a Gaussian function $\mathcal{N}(\mu,\Sigma)$, which induces the approximate Hamiltonian $\widetilde{\scripted{H}}(\theta,\scripted{v})$ in the form of (\ref{eq:approx_Hamiltonian}) and the closed--form trajectory (\ref{eq:ellipse}). If the mapping (\ref{eq:ellipse}), which is self--reverse with unity Jacobian, is used to compute the next proposal $\xi$, while the density $\scripted{q}_\scripted{r}(\scripted{r}\mid\theta,\scripted{v})$ satisfies condition (\ref{eq:r_symmetry}), then the approximate Hamiltonian term is also preserved, i.e.
\begin{equation*}
	\psi(\theta)\scripted{q}_\scripted{v}(\scripted{v}\mid\theta) = \psi(\xi)\scripted{q}_\scripted{v}(\scripted{w}\mid\xi),
\end{equation*}
and the acceptance probability will become 
\begin{equation*}
	\alpha(\theta,\scripted{V}) = \MN{1,\frac{\psi(\xi)\scripted{L}(\xi)\scripted{q}_\scripted{v}(\scripted{w}\mid\xi)}{\psi(\theta)\scripted{L}(\theta)\scripted{q}_\scripted{v}(\scripted{v}\mid\xi)}} = \MN{1,\frac{\scripted{L}(\xi)}{\scripted{L}(\theta)}}.
\end{equation*} 
This acceptance probability can be closer to unity, in comparison with plain Metropolis--Hastings, if $\scripted{L}(\theta)$ is a fairly flat likelihood function, i.e. $\psi(\theta)$ is a good approximation to $\pi(\theta)$. 

The proposal (\ref{eq:ellipse}), which originally presented by \cite{Neal1999}, was developed further by \cite{Murray2010} to recursively generate the fictional integration time $\scripted{r}^\star$ using slice sampling. Specifically, \cite{Murray2010} introduced an auxiliary variable, slightly different from conventional slice sampling, as $
	\scripted{h}\mid\theta\sim\mathcal{U}[0,\scripted{L}(\theta)]$, so that the joint density of $\theta,\scripted{v}$ and $\scripted{h}$ becomes
\begin{equation*}
	\p{}{\theta,\scripted{v},\scripted{h}} = \psi(\theta)\scripted{q}_\scripted{v}(\scripted{v}\mid\theta) = \exp(-\widetilde{\scripted{H}}(\theta,\scripted{v})) ;\quad\text{subjected to }\scripted{h}\le\scripted{L}(\theta).
\end{equation*}
We can draw samples from this density using the so called elliptical slice sampling algorithm given by \cite{Murray2010}, which essentially is the following Gibbs sampling scheme
\begin{align}
	\scripted{h}\mid\theta,\scripted{v} &\sim\mathcal{U}[0,\scripted{L}(\theta)],\nonumber \\
	\scripted{v}\mid\theta,\scripted{h} &\sim\mathcal{N}(0,\Sigma^{-1}),\nonumber\\
	\theta,\scripted{v}\mid\scripted{h} &\sim \exp(-\widetilde{\scripted{H}}(\theta,\scripted{v}));\quad\text{subjected to }\scripted{h}\le\scripted{L}(\theta).\label{eq:elliptical_slice}
\end{align}
To sample from (\ref{eq:elliptical_slice}), \cite{Murray2010} compute the proposal $\mat{\xi,&\scripted{w}}$ using (\ref{eq:ellipse}) and generate $\scripted{r}^\star\sim\scripted{q}^\star_\scripted{r}(\scripted{r}\mid\theta)$ using a simplified version of algorithms (\ref{alg:expand_slice}), where the slice is defined as 
\begin{equation*}
	\scripted{S}\defeq\{\mat{\theta,&\scripted{v}} : \scripted{L}(\theta)\ge\scripted{h}\}.
\end{equation*} 
See \citep[figure 2]{Murray2010} or algorithm (\ref{alg:HSS}) for details. As with slice sampling, the acceptance probability is unity again since the abstract density $\scripted{q}^\star_\scripted{r}(\scripted{r}\mid\theta)$ is symmetrical as in (\ref{eq:IPG_symmetry}) and the mapping (\ref{eq:ellipse}) preserves the Hamiltonian term $\widetilde{\scripted{H}}(\theta,\scripted{v})$ that appears in the target (\ref{eq:elliptical_slice}) while also results in a unity Jacobian. This explanation, which is rather straightforward in comparison with the original texts, is another benefit of the general framework provided in theorem (\ref{thm:MH}).

\cite{Nishihara2014a} provide a more general way to reformulate any generic target $\pi(\theta)$ in the form of (\ref{eq:tilde_pi}) by letting $\psi(\theta)$ be the best approximation to the target $\pi(\theta)$ and simply define $\scripted{L}(\theta)\defeq {\pi(\theta)}/{\psi(\theta)}$. To ensure the stability of the artificial function $\scripted{L}(\theta)$ in the tail--region of the target density, \cite{Nishihara2014a} also propose to construct $\psi(\theta)$ as a t--distribution given by
\begin{equation}\label{eq:t_dist}
	\psi(\theta) \defeq\int_0^\infty\phi(\theta\mid s)\mathsf{IG}(s : \frac{n}{2},\frac{n}{2})\dif s,
\end{equation} where $n$ is a chosen degree of freedom, $\mathsf{IG}(s : a,b)$ stands for the inverse--gamma density and $\phi(\theta\mid s)$ is the Gaussian density $\mathcal{N}(\mu,s\Sigma)$ so that for any given parameter $s$, elliptical slice sampling is still applicable to sampling from the density $\p{}{\theta\mid s}$. Accordingly, an efficient Gibbs sampling scheme for the joint density $\p{}{\theta,s}$ is also given in \citep[algorithm 2]{Nishihara2014a}, which is completely autonomous given that the free parameters $n,\mu,\Sigma$ are tuned by expectation maximisation \citep[algorithm 4]{Liu1995,Nishihara2014a}.

Taking advantage of the general framework, we can propose a simplification to the elliptical slice sampling algorithm in (\ref{eq:elliptical_slice}) that avoids the need of constructing and stabilising an artificial likelihood function $\scripted{L}(\theta)$ as follows. First we introduce the auxiliary slice variable differently as 
\begin{equation*}
	\scripted{h}\mid\theta,\scripted{v} \sim\mathcal{U}[0,\exp(-\scripted{H}(\theta,\scripted{v}))];\quad\scripted{H}(\theta,\scripted{v}) \defeq -\log\left( \pi(\theta)\scripted{q}_\scripted{v}(\scripted{v}\mid\theta) \right),
\end{equation*} which leads to $\pi(\theta)$ being a marginal density of a different joint density as follows $$\p{}{\theta,\scripted{v},\scripted{h}} = \mathbb{1}_\Upsilon(\theta,\scripted{v},\scripted{h});\quad\Upsilon\defeq\{\theta,\scripted{v},\scripted{h} : \log\scripted{h}\le-\scripted{H}(\theta,\scripted{v})\}.$$
We can sample from this joint density using the following Gibbs sampling scheme
\begin{align}
	\scripted{h}\mid\theta,\scripted{v} &\sim\mathcal{U}[0,\exp(-\scripted{H}(\theta,\scripted{v}))],\nonumber \\
	\scripted{v}\mid\theta,\scripted{h} &\sim\mathbb{1}_{\scripted{S}^\scripted{v}}(\scripted{v}); &\scripted{S}^\scripted{v} &\defeq\{\scripted{v} : \log\scripted{h}\le-\scripted{H}(\theta,\scripted{v})\},\label{eq:Hamiltonian_slice1}\\
	\theta,\scripted{v}\mid\scripted{h} &\sim\mathbb{1}_\scripted{S}(\theta,\scripted{v}); &\scripted{S} &\defeq \{\mat{\theta,&\scripted{v}} : \log\scripted{h}\le-\scripted{H}(\theta,\scripted{v})\}.\label{eq:Hamiltonian_slice2}
\end{align} 
We can draw exact i.i.d. samples from (\ref{eq:Hamiltonian_slice1}) as described in appendix (\ref{appx:uniform_ellipsoid}), while the joint conditional density (\ref{eq:Hamiltonian_slice2}) can be solved using the mapping (\ref{eq:ellipse}) where the random integration time $\scripted{r}^\star\sim\scripted{q}^\star_\scripted{r}(\scripted{r}\mid\theta)$ is recursively generated as in algorithm (\ref{alg:HSS}), with the slice $\scripted{S}$ defined differently as in (\ref{eq:Hamiltonian_slice2}). If $\widetilde{\scripted{H}}(\theta,\scripted{v})$ is a good approximation to the true Hamiltonian $\scripted{H}(\theta,\scripted{v})$, then $\scripted{H}(\theta,\scripted{v})$ will be approximately invariant along the elliptical trajectory (\ref{eq:ellipse}). Therefore, this trajectory may carry the proposal $\xi$ to a far distance from $\theta$, without falling out of the slice $\scripted{S}$.
\begin{remark}
	The No--U--Turn sampler given by \cite{Hoffman2014} is, at its core, an analogous version of the Gibbs sampling scheme (\ref{eq:Hamiltonian_slice2}) where the elliptical trajectory (\ref{eq:ellipse}) is replaced with the discretised solution of the original Hamiltonian equations (\ref{eq:Hamiltonian}). Instead of using the ``stepping out" approach seen in algorithm (\ref{alg:expand_slice}) to explore the range of the slice $\scripted{S}$ given in (\ref{eq:Hamiltonian_slice2}), \cite{Hoffman2014} designed a discretised adaptation of the ``doubling" procedure, which is yet another special case of (\ref{eq:artificial_posterior2}) originally given by \cite{Neal2003}. 
\end{remark}
\begin{algorithm}
	\caption{Hamiltonian Slice Sampling}
	\label{alg:HSS}
	First, we set $ a_1 = u-2\pi ;\;	b_1 = u $ for some $ u\sim\mathcal{U}[0,2\pi] .$ Then, we adaptively generate $[a_n,b_n]$, starting with $n=1$, by shrinking $[a_1,b_1]$ as follows
	\begin{enumerate}
		\item Sample $ \scripted{r}_n\sim\mathcal{U}[a_n,b_n] $ and compute $\scripted{T}:\mat{\theta_k, & \scripted{v}_k, & \scripted{r}_n} \mapsto \mat{\xi_n, & \scripted{w}_n, & \scripted{s}_n}$ using (\ref{eq:ellipse}) \label{st:sampling_i}
		\item Set $ \mat{\theta_{k+1}, & \scripted{v}_{k+1}} \gets \mat{\xi_n, & \scripted{w}_n}$ if $ \mat{\xi_n, & \scripted{w}_n}\in\scripted{S}$ and terminate the routine
		\item Otherwise, repeat step (\ref{st:sampling_i}) with $ [a_{n+1},b_{n+1}] $ generated by
		\begin{enumerate}
			\item If $ \scripted{r}_{n}\ge 0 $, then set $ b_{n+1} =\scripted{r}_{n};\;	a_{n+1} = a_{n} $
			\item Otherwise, set $	b_{n+1} = b_{n};\;	a_{n+1} =\scripted{r}_{n} $
		\end{enumerate}
	\end{enumerate}
\end{algorithm} 

The so called Hamiltonian slice sampling scheme (\ref{eq:Hamiltonian_slice2}) is perhaps a simpler alternative to the generalised elliptical slice sampling algorithm given by \cite{Nishihara2014a}. In comparison, both algorithms need to estimate $\mu,\Sigma$ while Hamiltonian slice sampling is naturally insensitive to heavy--tailed densities, hence no need to use the t--distribution (\ref{eq:t_dist}). These algorithms remain valid if $\Sigma$ is the inverse observed information matrix as suggested by \cite{Girolami2011}, if this choice is indeed advantageous in some applications. As we will see in section (\ref{sec:PHMC}), these Hamiltonian slice sampling schemes remain applicable even when $\pi(\theta)$ is intractable but can be unbiasedly estimated, e.g. using particle filtering. We first give an introductory overview of the pseudo marginal sampling approaches in the next section.
\subsection{Pseudo Marginal Metropolis--Hastings}
For some Bayesian models, the target density involves a likelihood function that is intractable but can be unbiasedly approximated using either importance sampling \citep{Herman1951,Richard2007} or particle filtering \citep{Gordon1993,Doucet2000} algorithms. One popular example of this model class is the state space model discussed in appendix (\ref{appx:ABC}). 

In general, assuming that we want to draw samples from an intractable density $\pi(\theta)$ which can be unbiasedly approximated by a tractable function $\widehat{\pi}(\theta,\mathbb{u})$ in the sense that 
\begin{equation}
  \E{\mathbb{u}\mid\theta}{\widehat{\pi}(\theta,\mathbb{u})}\defeq\int\widehat{\pi}(\theta,\mathbb{u})\gamma_{\mathbb{u}}(\mathbb{u}\mid\theta)\dif\mathbb{u} = \pi(\theta), 
\end{equation} where $\mathbb{u}\mid\theta\sim\gamma_{\mathbb{u}}(\mathbb{u}\mid\theta)$ is some possibly intractable density that we can simulate from. To obtain $\theta_k\sim\pi(\cdot)$, \cite{Andrieu2009a} propose to sample instead from the joint density $\varphi(\theta,\mathbb{u})$, which admits $\pi(\theta)$ as its marginal density, constructed as $\varphi(\theta,\mathbb{u}) = \widehat{\pi}(\theta,\mathbb{u})\gamma_{\mathbb{u}}(\mathbb{u}\mid\theta)$.

Now we can actually use the conventional Metropolis--Hastings algorithm \citep{Roberts2004a} to sample from $\varphi(\theta,\mathbb{u})$ with the following proposal density
\begin{equation}
  \gamma(\xi,\mathbb{w}\mid\theta,\mathbb{u})\defeq\gamma_{\mathbb{u}}(\mathbb{w}\mid\xi) \gamma_\theta(\xi\mid\theta),
\end{equation} which leads to the acceptance probability written in conventional Metropolis notation as follows
\begin{equation}\label{eq:PMMH_acc}
	\begin{aligned}
		\alpha(\xi,\mathbb{w}\mid\theta,\mathbb{u}) &\defeq \MN{1, \frac{\gamma(\theta,\mathbb{u}\mid\xi,\mathbb{w})}{\gamma(\xi,\mathbb{w}\mid\theta,\mathbb{u})}\frac{\varphi(\xi,\mathbb{w})}{\varphi(\theta,\mathbb{u})} } = \MN{1, \frac{\gamma_\theta(\theta\mid\xi)}{\gamma_\theta(\xi\mid\theta)}\frac{\widehat{\pi}(\xi,\mathbb{w})}{\widehat{\pi}(\theta,\mathbb{u})} }.
	\end{aligned}
\end{equation} 

Operationally, this sampling scheme appears as if we simply replace the intractable density $\pi(\theta)$ with an unbiased estimator $\widehat{\pi}(\theta,\mathbb{u})$ in a plain Metropolis--Hastings algorithm. The user choice of the proposal density $\gamma_\theta(\xi\mid\theta)$ is therefore becomes the target of intense studies and innovations, see \citep{Dahlin2016b,Lee2008} for examples.

When the proposal density $\gamma_\theta(\xi\mid\theta)$ is symmetrical, the acceptance rate will be simply $$\alpha(\xi,\mathbb{w}\mid\theta,\mathbb{u}) = \MN{1, \frac{\widehat{\pi}(\xi,\mathbb{w})}{\widehat{\pi}(\theta,\mathbb{u})} },$$ which is a noisy approximation to the true Metropolis acceptance rate (\ref{eq:Metropolis_acc}) due to the random variation in $\mathbb{v,w}$. This effect can cause persistent rejections, which is often called sticking, in pseudo marginal random walk Metropolis, even when $\pi(\xi)\approxeq\pi(\theta)$. Therefore to avoid sticking, it is immediately clear that we need to minimise the variance of $\widehat{\pi}(\theta,\mathbb{u})$ with respect to the random variation of $\mathbb{u}\mid\theta$. This objective can be achieved by simply increasing the number of samples required to compute the estimator $\widehat{\pi}(\theta,\mathbb{u})$ at the cost of linearly increasing the computation time. 

\cite{Pitt2012,Doucet2015a} suggest that the number of required samples is chosen such that $\V{\mathbb{u}\mid\theta}{\log\widehat{\pi}(\theta,\mathbb{u})}\in[0.5^2,1.5^2]$ to concurrently minimise both the factor of inefficiency $\tau_{\hat{f}}$ in (\ref{eq:convergence}) and the computation time. Also under different assumptions and taking into account the diffusion speed of the resulting chain $\{\theta_k\}_{k=1}^M$ measured by $\mathsf{E}\norm{\theta_{k+1}-\theta_k}^2$, instead of the factor $\tau_{\hat{f}}$, \cite{Sherlock2015} conclude that the optimal value of $\V{\mathbb{u}\mid\theta}{\log\widehat{\pi}(\theta,\mathbb{u})}\approx 3.3$, while the corresponding optimal acceptance rate of pseudo marginal random walk Metropolis algorithm is also found to be approximately $7\%$, instead of $23.4\%$ as with conventional Metropolis. 

Additionally, the performance of the pseudo marginal Metropolis--Hastings approach can be significantly improved by either using specialised design for the estimator $\widehat{\pi}(\theta,\mathbb{u})$, e.g. using auxiliary particle filter \citep{Pitt1999,Pitt2012}, or creating artificial correlation between $\mathbb{u}$ and $\mathbb{w}$ in consecutive iterations by altering the proposal density $\gamma_{\mathbb{u}}(\mathbb{w}\mid\xi,\mathbb{u})$ as studied by \cite{Deligiannidis2015,Dahlin2015b,Jacob2016}.

The random walk Metropolis proposal (\ref{eq:Metropolis}) is not only prone to sticking in the pseudo marginal context but also inherently inefficient when the number of dimensions $n_\theta$ increase. One preferable alternative is the discretised Langevin diffusion proposal given in (\ref{eq:Langevin}), which takes advantage of the log-likelihood gradients to guide the proposal to region of higher density. 

The required gradients can be approximated using a varieties of methods as studied by \cite{Dahlin2013,Nemeth2014a}. The observed information matrix can also be approximated at the same time with the gradient vector if we chose to set the matrix $\Sigma$ in (\ref{eq:Langevin}) as the inverse observed information matrix in each iteration. 

The main motivation for modulating the Langevin proposal (\ref{eq:Langevin}) using the observed information matrix is that the resulting algorithm is scale invariant, e.g. the resulting Markov chain has larger jumps in flat regions and smaller jumps in sharp regions of the target density as seen in \citep[figure 1]{Girolami2011}, which can be advantageous during burn--in.

\cite{Poyiadjis2011,Dahlin2015,Dahlin2015a,Nemeth2015} study methods to approximate the observed information matrix along with the gradient vectors and find that the diffusion speed of the chain is improved during both burn--in and stationary phases. However, \cite{Nemeth2014} show that this computational efficiency gain strongly depends on the accuracy of the estimation of the gradients, especially when $n_\theta$ increases. Therefore, we present a derivative--free approach in the next section to complement these aforementioned pseudo marginal approaches. 

\subsection{Pseudo Marginal Hamiltonian Slice Sampling}\label{sec:PHMC}
Now if the pseudo marginal Markov chain is in stationary setting or if the target density is approximately Gaussian then using the global covariance matrix of $\pi(\theta)$ in constructing the Langevin proposal (\ref{eq:Langevin}) could enable the same level of computational efficiency as with the local information matrix. This global covariance matrix can simply be approximated either by using the complete history of the Markov chain, e.g. as in (\ref{eq:AM_cov}) according to \cite{Andrieu2008}, or the population of parallel Markov chains available in a sequential Monte Carlo context described by \cite{Chopin2013}. This global covariance matrix can be a computationally cheap and statistically efficient alternative to the local information matrix.
 
Since discretised Langevin proposal is indeed a special case of Hamiltonian Monte Carlo sampling, and if the motivation given above for using the global covariance matrix is appropriate, then we should also take advantage of the elliptical Hamiltonian trajectory (\ref{eq:ellipse}) to construct the proposal in this pseudo marginal setting. First, we introduce the fictional momentum vector
\begin{equation*}
	\scripted{v}\mid\theta,\mathbb{u}\sim\gamma_\scripted{v}(\scripted{v}\mid\theta)\defeq\mathcal{N}\left(0,\Sigma^{-1}\right),
\end{equation*} as with Hamiltonian slice sampling and form the joint density 
\begin{equation*}
	\varphi(\theta,\scripted{v},\mathbb{u}) \defeq \gamma_\scripted{v}(\scripted{v}\mid\theta)\gamma_\mathbb{u}(\mathbb{u}\mid\theta)\widehat{\pi}(\theta,\mathbb{u}),
\end{equation*} which admits $\pi(\theta,\scripted{v})\defeq\pi(\theta)\gamma_\scripted{v}(\scripted{v}\mid\theta)$, and therefore also $\pi(\theta)$, as its marginal densities. 

Second, we introduce the slice sampling auxiliary variable as follows
\begin{equation*}
	\scripted{h}\mid\theta,\scripted{v},\mathbb{u}\sim\mathcal{U}[0,\gamma_\scripted{v}(\scripted{v}\mid\theta)\widehat{\pi}(\theta,\mathbb{u})],
\end{equation*} so that the joint density of $\mat{\theta,&\scripted{v},&\mathbb{u},&\scripted{h}}$ becomes
\begin{equation*}
	\varphi(\theta,\scripted{v},\mathbb{u},\scripted{h}) = \gamma_\mathbb{u}(\mathbb{u}\mid\theta);\quad\text{ subjected to } \scripted{h}\le\gamma_\scripted{v}(\scripted{v}\mid\theta)\widehat{\pi}(\theta,\mathbb{u}).
\end{equation*}
We can sample from this joint density using the following Gibbs sampling scheme
\begin{align}
	\scripted{h}\mid\theta,\scripted{v},\mathbb{u}&\sim\mathcal{U}[0,\gamma_\scripted{v}(\scripted{v}\mid\theta)\widehat{\pi}(\theta,\mathbb{u})],\\
	\scripted{v}\mid\theta,\scripted{h},\mathbb{u} &\sim\mathbb{1}_{\scripted{S}^\scripted{v}}(\scripted{v});\quad\scripted{S}^\scripted{v}\defeq\{\scripted{v} : \scripted{h}\le\gamma_\scripted{v}(\scripted{v}\mid\theta)\widehat{\pi}(\theta,\mathbb{u}) \},\\
	\theta,\scripted{v},\mathbb{u}\mid\scripted{h} &\sim \gamma_\mathbb{u}(\mathbb{u}\mid\theta)\mathbb{1}_\scripted{S}(\theta,\scripted{v},\mathbb{u});\quad\scripted{S}\defeq\{\mat{\theta,&\scripted{v},&\mathbb{u}} : \scripted{h}\le\gamma_\scripted{v}(\scripted{v}\mid\theta)\widehat{\pi}(\theta,\mathbb{u}) \}\label{eq:pseudo_Hamiltonian}.
\end{align}
While the other conditional densities are easy to simulate, we can sample from (\ref{eq:pseudo_Hamiltonian}) by using the trajectory (\ref{eq:ellipse}) to compute $\mat{\xi,&\scripted{w}}$ with the random integration time simulated from some symmetrical density $\scripted{r}\sim\scripted{q}_\scripted{r}(\cdot)$ as in (\ref{eq:r_symmetry}), which results in the following acceptance probability
\begin{equation}\label{eq:PHMC_acc}
	\begin{aligned}\alpha(\xi,\scripted{w},\mathbb{w}\mid\theta,\scripted{v},\mathbb{u}) &= \MN{1,\frac{\gamma_\mathbb{u}(\mathbb{w}\mid\xi)\mathbb{1}_\scripted{S}(\xi,\scripted{w},\mathbb{w})}{\gamma_\mathbb{u}(\mathbb{u}\mid\theta)\mathbb{1}_\scripted{S}(\theta,\scripted{v},\mathbb{u})}\frac{\gamma_\mathbb{u}(\mathbb{u}\mid\theta)}{\gamma_\mathbb{u}(\mathbb{w}\mid\xi)} } = \mathbb{1}_\scripted{S}(\xi,\scripted{w},\mathbb{w}).
	\end{aligned}
\end{equation}

Even though we cannot apply a recursive sampling scheme in the pseudo marginal context, we can still use $\scripted{h}$ to scale the proposal for $\scripted{r}\sim\scripted{q}_\scripted{r}(\scripted{r}\mid\scripted{h})$ so that this acceptance rate is equal to one with higher probability. This is because $\psi(\theta)\approx\pi(\theta)\approx\widehat{\pi}(\theta,\mathbb{u})$ and since the term ${\gamma_\scripted{v}(\scripted{v}\mid\theta)\psi(\theta)}$ is preserved by the Hamiltonian mapping (\ref{eq:ellipse}), hence the term $\gamma_\scripted{v}(\scripted{v}\mid\theta)\widehat{\pi}(\theta,\mathbb{u})$ should also remain approximately level along the elliptical Hamiltonian trajectory. Therefore we should propose larger $\scripted{r}$ when $\scripted{h}$ is small and vice versa, e.g. $\scripted{r}\sim\mathcal{N}(0,\sigma^2/\scripted{h}^\zeta);\;\sigma,\zeta>0$ or a truncated Gaussian with threshold equal to $\scripted{h}/\gamma_\scripted{v}(\scripted{v}\mid\theta)\widehat{\pi}(\theta,\mathbb{u})$.

While previous approaches are carefully designed so that they only require the computation of $\widehat{\pi}(\theta,\mathbb{u})$, we note that any slice sampling scheme can also be employed to draw samples from $\varphi(\theta,\mathbb{u})$ as done by \cite{Murray2016}. Furthermore, discretised Hamiltonian sampling is also applicable if $\varphi(\theta,\mathbb{u})$ is designed to be a smooth density either by not performing resampling in particle filtering, as done by \cite{Lindsten2016}, or by following \cite{Malik2011}.
\section{Discussion}

Since we have provided an unified framework for many Markov chain Monte Carlo algorithms, it can now be natural to think of a general convergence theory for all these sampling approaches based on studying the properties of the proposal density of $\scripted{V}$ and the mapping $\scripted{T}$. While we have restricted our attention to algorithms that corresponds to a self--reverse mapping, the original framework given by \cite{Green1995} also allows for mappings that are not self--inverse. It is an open question whether this framework therefore can also be generalised to study nonreversible Markov chain samplers as studied by \cite{Diaconis2000,Sohl-Dickstein2014,Bierkens2015}.

On the practical aspects of Markov chain sampling, we have highlighted the fact that both slice sampling and Metropolis sampling are special cases of a single theorem. However, slice sampling additionally has a recursive proposal mechanism that is capable of automatically tuning the scale of the generated proposals to guarantee a new sample in each iteration. We have also presented new ways to perform slice sampling, e.g. on a random direction or an hyper--ellipse, which are based on embedding different samplers inside various Gibbs sampling schemes. 

Among the samplers presented in this paper, only directional slice sampling, Hamiltonian slice sampling and discretised Hamiltonian sampling can be used as general purpose ``black--box" samplers. Directional slice sampling is computationally cheap since it does not require matrix computation as other samplers. On the other hand, discretised Hamiltonian sampling can be computationally expensive but is popular and well developed by \cite{Carpenter2016}. Hamiltonian slice sampling can be a mid--way approach that does not require gradient computation while can still make use of the knowledge of the covariance matrix $\Sigma$, if available. This latter approach is also applicable in applications with pseudo marginal density. Additional simulations of these ``black--box" samplers on practical applications is available in the supplementary material.

Finally, when the target density has a high number of dimensions, the usefulness of the covariance matrix $\Sigma$ is limited by the computational cost of associated matrix operations. More importantly, we also currently lack advanced algorithms to estimate this matrix in a recursive manner, in the sense that the estimation is incrementally repeated when new samples arrive. Meanwhile, an optimisation--based approach to covariance estimation can induce significant computational cost in sampling algorithms that require an estimation of $\Sigma$, e.g. see \citep[figure 4]{Nishihara2014a} for a simulation. Therefore, more research in covariance estimation in the context of Markov chain Monte Carlo sampling is necessary to enhance the efficiency of sampling methods such as Hamiltonian slice sampling or Metropolis adjusted Langevin algorithms. One notable contribution in this direction is the application of quasi--Newton approximations from the optimisation literature to Markov chain Monte Carlo by \cite{Zhang2011,Dahlin2015a}.

%\section*{Acknowledgement}
%We thank Christopher Drovandi, Brett Ninness and  Radford M. Neal for valuable discussions.

\section*{Supplementary material}
\label{SM}

Supplementary material is not yet available %at \Bka\ online.

\appendix
%\appendixone
\section*{Appendix}
\subsection{Metropolis--Hastings without Self--Reverse Mapping}\label{appx:Green}
The invariance condition (\ref{eq:invariant}) can be guaranteed if the following reversibility condition holds
\begin{equation}\label{eq:reversibility}
	\scripted{K}(\dif\xi\mid\theta)\wp(\dif\theta) = \scripted{K}(\dif\theta\mid\xi)\wp(\dif\xi);\quad\wp(\dif\theta)\defeq\pi(\theta)\dif\theta ,
\end{equation} since this will results in 
\begin{equation*}
	\int_\theta\scripted{K}(\dif\xi\mid\theta)\wp(\dif\theta) =\int_\theta\scripted{K}(\dif\theta\mid\xi)\wp(\dif\xi) = \wp(\dif\xi).
\end{equation*}
Since condition (\ref{eq:reversibility}) is trivial when $\xi = \theta$, we only consider the case when $\xi\neq\theta$, which leads to an equivalent expression of reversibility as follows
\begin{equation*}
	\alpha(\xi\mid\theta)\scripted{Q}(\dif\xi\mid\theta)\wp(\dif\theta) = \alpha(\theta\mid\xi)\scripted{Q}(\dif\theta\mid\xi)\wp(\dif\xi),
\end{equation*} where $\scripted{Q}(\dif\xi\mid\theta)$ denotes the proposal distribution of the Metropolis--Hastings algorithm. 

Following \cite{Peskun1973,Green1995}, we can see that the reversibility condition is optimally satisfied with the following acceptance probability
\begin{equation}\label{eq:general_acc}
	\alpha(\xi\mid\theta)\defeq\MN{1,\frac{\scripted{Q}(\dif\theta\mid\xi)\wp(\dif\xi)}{\scripted{Q}(\dif\xi\mid\theta)\wp(\dif\theta)}}.
\end{equation}

According to \cite{Green1995}, we can denote $\lambda(\cdot)$ as the $(n_\theta+n_\scripted{V})$--dimensional Lebesgue measure and define a symmetric measure on the product space $\theta\in\scripted{A}\subset\mathbb{R}^{n_\theta}\times\xi\in\scripted{B}\subset\mathbb{R}^{n_\theta}$ as follows
\begin{equation*}
	\kappa(\scripted{A}\times\scripted{B}) = \kappa(\scripted{B}\times\scripted{A}) \defeq \lambda\{\mat{\theta,&\scripted{V}} : \theta\in\scripted{A},\xi\in\scripted{B}\};\quad (\xi,\scripted{W})\defeq\scripted{T}(\theta,\scripted{V}),
\end{equation*} which induces the following related densities 
\begin{equation}\label{eq:densities}
	\begin{aligned}
	\scripted{Q}(\dif\theta\mid\xi)\wp(\dif\xi) &= \pi(\xi,\scripted{W})\vmat{J_\scripted{T}(\theta,\scripted{V})}\lambda(\dif\theta\dif\scripted{V}) \\
	\scripted{Q}(\dif\xi\mid\theta)\wp(\dif\theta) &= \pi(\theta,\scripted{V})\lambda(\dif\theta\dif\scripted{V}).
	\end{aligned}
\end{equation} Finally substituting (\ref{eq:densities}) into (\ref{eq:general_acc}) will result in the acceptance probability given in theorem (\ref{thm:MH}) without requiring $\scripted{T}$ to be a self--reverse mapping.

\subsection{Elliptical Hamiltonian Dynamics}\label{appx:Hamiltonian}

Without loss of generality, we can let $\mu=0$ and define the Hamiltonian term as follows
\begin{equation}\label{eq:simple_Hamiltonian}
	\scripted{H}(\theta,\scripted{v}) = \frac{1}{2}\left(\theta^\mathsf{T}\Sigma^{-1}\theta + \scripted{v}^\mathsf{T}\Sigma\scripted{v}\right),
\end{equation}
the exact Hamiltonian dynamics can be derived using the following transforms of variables. Specifically, let the covariance matrix $\Sigma$ have the following Cholesky decomposition $$\Sigma = \scripted{M}^\mathsf{T}\scripted{M} \Rightarrow \Sigma^{-1} = \scripted{M}^{-1}\left(\scripted{M}^\mathsf{T}\right)^{-1},$$ and using the following linear transforms $\theta\defeq\scripted{M}^\mathsf{T}\scripted{x},\;\scripted{v}\defeq\scripted{M}^{-1}\scripted{y}$, we can rewrite the Hamiltonian as 
\begin{equation}\label{eq:simple_Hamiltonian2}
	\scripted{H}(\scripted{x},\scripted{y}) = \frac{1}{2}\left(\scripted{x}^\mathsf{T}\scripted{x} + \scripted{y}^\mathsf{T}\scripted{y}\right),
\end{equation} which leads to the following closed--form solution \citep{Pakman2014}
\begin{equation*}
	\begin{aligned}
		\scripted{x}(t) &= \scripted{x}(0)\cos(t) + \scripted{y}(0)\sin(t),\\
		\scripted{y}(t) &= \scripted{y}(0)\cos(t) - \scripted{x}(0)\sin(t).
	\end{aligned}
\end{equation*} 
According to \cite[section 4.1]{Neal2012}, this trajectory is equivalent with the following tracjectory in the original domain
\begin{equation*}
	\begin{aligned}
		\theta(t) &= \theta(0)\cos(t) + \Sigma\scripted{v}(0)\sin(t),\\
		\scripted{v}(t) &= \scripted{v}(0)\cos(t) - \Sigma^{-1}\theta(0)\sin(t).
	\end{aligned}
\end{equation*} When $\mu\neq 0$, we simply need to shift the coordinates of $\theta$ by replacing $\theta$ with $\theta-\mu$ to retrieve the corresponding trajectory.

\subsection{Uniform Sampling From the Volume of an Ellipsoids}\label{appx:uniform_ellipsoid}

Given that the momentum vector is distributed as $ \scripted{v}\mid\theta\sim\mathcal{N}(0,\Sigma^{-1}) $, the indicator condition of the slice, i.e. $\log\scripted{h}\le-\scripted{H}(\theta,\scripted{v})$, is also equivalent to 
\begin{equation}\label{eq:hyperellipse}
	\frac{\scripted{v}^\mathsf{T}\Sigma\scripted{v}}{\rho}\le1;\quad\rho\defeq -\left( 2\log \frac{h}{\pi(\theta)} + n_\theta\log(2\mathbb{\bbpi}) + \log|\Sigma^{-1}|\right),
\end{equation} so the density given by (\ref{eq:Hamiltonian_slice1}) is uniform on the ellipsoid defined by (\ref{eq:hyperellipse}). We can draw samples from this density by first drawing some vector $\scripted{y}$ uniformly distributed in the unit ball as follows
\begin{equation*}
	\scripted{y}=\frac{\scripted{z}}{\norm{\scripted{z}}}u^{1/n_\theta};\quad u\sim\mathcal{U}[0,1];\quad\scripted{z}\sim\mathcal{N}(0,\mathbb{I}_{n_\theta}).
\end{equation*}
The rationale for this is because $\p{}{\scripted{y}\in\scripted{Ball}(0,r)} = \p{}{u^{1/n_\theta} \le r} = r^{n_\theta} $ and $\mathsf{volume}[\scripted{Ball}(0,r)]\propto r^{n_\theta}$. Hence, the density of $\scripted{y}$ will be constant with respect to the radius $r$, while the random Gaussian vector ${\scripted{z}}/{\norm{\scripted{z}}}$ is isotropic and therefore uniformly distributed on the unit sphere. Finally, we map $\scripted{y}$ to the desired ellipsoid by $\scripted{v} = \scripted{M}^{-1}\scripted{y}/\sqrt{\rho}$ where $\scripted{M}$ is previously defined as the Cholesky decomposition of $\Sigma$. The vector $\scripted{v}$ will also be uniformly distributed since $\p{\scripted{v}}{\scripted{v}} = \p{\scripted{y}}{\scripted{Mv}}\vmat{\scripted{M}} $ is also a uniform density in $\scripted{v}$.
\subsection{Pseudo Marginal and Approximate Bayesian Computation in State Space Models}\label{appx:ABC}
Time series are statistical data that appears in science, engineering, health and economics as time sequences of random variables $y_t\in\mathbb{R}^{n_y}$  that are jointly denoted as $ Y_T \defeq \{y_t\}_{t=1}^T$. The natural ordering in time often implies a statistical dependency of each data point $y_t$ in the series upon its prior history $Y_{t-1}$. This dependency is central to modelling time series data and also predicting its behaviours in the near future time steps. Ordinarily, data points that are closer together in time will be more correlated than data that comes from further apart. One way to model this temporal dependency is to base the evolution of the time series upon a Markov chain of the state variables $x_t\in\mathbb{R}^{n_x}$ as follows
\begin{equation}\label{eq:SSM}
\begin{aligned}
	  x_{t+1}\mid x_t &\sim f_\theta(x_{t+1}\mid x_t);\quad x_1\sim f_\theta(x_1), \\
	  y_t\mid x_t &\sim g_\theta(y_t\mid x_t);\quad \theta\sim\psi(\theta),
\end{aligned}
\end{equation} where $\theta\in\mathbb{R}^{n_\theta}$ denotes some static parameters that characterise the densities $f_\theta, g_\theta$. 

In state space model, we can recognise that the target $\pi(\theta)\defeq\p{}{\theta\mid Y_T}$ is a marginal density of a larger density $ \pi(\theta,X_T)\defeq\p{}{\theta,X_T\mid Y_T} $, but we may not be able to sample directly from this density due to large number of dimensions even when the constituent densities $f_\theta, g_\theta$ are all tractable. Furthermore, we are sometimes only interested in sampling $\theta\sim\pi(\theta)$, which has a significantly smaller dimensionality and can also be computed up to a normalising constant as follows\begin{equation}
	\pi(\theta)\defeq\p{}{\theta\mid Y_T}\propto\psi(\theta)\p{\theta}{Y_T},
	\end{equation} where $\p{\theta}{Y_T}$ can be approximated by $\phat{\theta,\mathbb{u}}{Y_T}$ using particle filtering as described by \cite{Doucet2000}.

In this scenario, the pseudo marginal Metropolis--Hastings algorithm described in section (\ref{sec:PHMC}) is directly applicable e.g. by letting $\mathbb{u}\sim\gamma_\mathbb{u}(\mathbb{u}\mid\theta)$ correspond to all the random numbers generated by the particle filter and $\widehat{\pi}(\theta,\mathbb{u})\defeq\psi(\theta)\phat{\theta,\mathbb{u}}{Y_T}$, where the construction of $\widehat{\pi}(\theta,\mathbb{u})$ can be arbitrarily intricate as long as it gives the best possible estimator for $\pi(\theta)$ because we do not need to compute $\gamma_\mathbb{u}(\mathbb{u}\mid\theta)$.

One popular example of Bayesian models with truly intractable posterior or likelihood function is a state space model with either intractable state transition density $f_\theta(x_{t+1}\mid x_t)$ or measurement density $g_\theta(y_t\mid x_t)$. However, if simulating random draws from $f_\theta$ and $g_\theta$ is still feasible, then we can construct an approximate model of (\ref{eq:SSM}) that avoids having to calculate these densities in order to perform inference on $\theta$. This is the motivation of approximate Bayesian computation (ABC) in state space model. According to \cite{Dahlin2015a}, this approximate model can be constructed by first creating an auxiliary data set $Z_T\defeq\{z_t\}_{t=1}^T$ using some tractable density $\nu_t\sim\phi(\cdot) $, with normalised variance $\V{\phi}{\nu} = 1$, as follows
\begin{equation*}
	z_t = \kappa(y_t) + \epsilon\nu_t;\quad\epsilon>0,
\end{equation*} where $\kappa(\cdot)$ is a deterministic function.
Then we have the following approximate state space model
\begin{align}
	x_{t+1},\widetilde{y}_{t+1}\mid x_t, \widetilde{y}_t &\sim f_\theta(x_{t+1}\mid x_t) g_\theta(\widetilde{y}_{t+1}\mid x_{t+1}),\\
	z_t\mid x_t, \widetilde{y}_t &\sim\phi\left(\frac{z_t-\kappa(\widetilde{y}_t)}{\epsilon}\right).\label{eq:approx_measurement}
\end{align}
Finally, we can use the bootstrap particle filter algorithm given by \cite{Gordon1993,Doucet2000} to draw samples from $\p{\theta}{X_T,\widetilde{Y}_T\mid Z_T},\;\widetilde{Y}_T\defeq\{\widetilde{y}_t\}_{t=1}^T,$ and estimate the likelihood function for the approximate model, which is $\p{\theta}{Z_T}$. 

The argument for this approach essentially is based on the intuition that $Z_T\to Y_T$ as $\epsilon\to 0$, or i.e. $Z_T$ becomes sufficiently informative about $\theta$ in comparison with the original data $Y_T$. However, the artificial measurement density (\ref{eq:approx_measurement}) will become arbitrarily informative as $\epsilon\to 0$, i.e. $\p{}{z_t\mid x_t, \widetilde{y}_t}$ becomes an arbitrarily sharp function, which makes conventional filtering methods, such as the bootstrap particle filter, perform poorly. This phenomenon can be seen in \citep[figure 1]{Dahlin2015a} where the error in approximating $\log\p{\theta}{Y_T}$, using a bootstrap estimate of $\log\p{\theta}{Z_T}$, grows sharply as $\epsilon$ approaches zero. 

One heuristic argument for choosing a good value for $\epsilon$ is to see that, according to the approximate model, we have 
\begin{equation*}
	\V{\theta}{\kappa(\widetilde{y}_t) - \kappa(y_t)\mid z_t} = \V{\theta}{\kappa(\widetilde{y}_t) - z_t + z_t - \kappa(y_t)\mid z_t} = 2\epsilon^2;\quad\forall z_t.
\end{equation*} 
Therefore if $\phi(\cdot)$ is Gaussian, then we do not need to generate $Z_T$ and simply rewrite (\ref{eq:approx_measurement}) as 
\begin{equation}
	y_t\mid x_t, \widetilde{y}_t \sim\phi\left(\frac{\kappa(y_t)-\kappa(\widetilde{y}_t)}{\epsilon\sqrt{2}}\right).\label{eq:approx_measurement2}
\end{equation}
This formulation is inline with the original approach given by \cite{Jasra2012} except that $\phi(\cdot)$ is chosen as a uniform density. Now if $\theta$ is close to the true parameter $\theta_\star$ then 
\begin{equation*}
	\V{\theta}{\kappa(\widetilde{y}_t) - \kappa(y_t)}\approx 2\V{\theta_\star}{\kappa(y_t)}.
\end{equation*} Hence it is reasonable to set $\epsilon\approx\sqrt{\V{\theta_\star}{\kappa(y_t)}}$ if an approximate of this constant is available. This argument seems to be supported by the simulation given in \citep[section 5.1]{Dahlin2015a}.

%\appendixtwo
%\section*{Appendix B}
\subsection{A Toy Example}\label{appx:toy_ex}
We now illustrate the performance of Metropolis sampling with the following example model, which is rich enough to also serve as a template for illustrating the performance of other sampling approaches, while also is simple enough to not interfere with the clarity of the present discussions. By using a simple model, we can also ensure that we are not incidentally favouring any one of the algorithms and we can simply expect all of them to be working at their best possible performance so we can intuitively judge the strength and weakness of each algorithm.
\begin{example}\label{ex:Metropolis}
We consider a toy example where the experiment data $Y_T\defeq\{\scripted{y}_t : t=1,2,\dots T\}$, are generated by the following model
	\begin{equation}\label{eq:ex_model}
		\scripted{y}_t = \scripted{x}_t + \epsilon_t;\quad\scripted{x}_t\sim\mathcal{N}(0,1);\quad\epsilon_t\sim\mathcal{N}(0,e^\upsilon);\quad\scripted{y}_t,\scripted{x}_t,\epsilon_t,\upsilon\in\mathbb{R} ,
	\end{equation} which is also equivalent to simply $\scripted{y}_t\sim\mathcal{N}(0,1+e^\upsilon)$. Hence, we can derive the posterior density of $\upsilon\mid Y_T$ as
	\begin{equation*}
		\pi(\upsilon)\defeq\p{}{\upsilon\mid Y_T} \propto \frac{\p{}{\upsilon}\exp\left(\frac{-1}{2\SqrBrk{1+e^\upsilon}}\sum_{t=1}^T\scripted{y}_t^2\right)}{\SqrBrk{1+e^\upsilon} ^\frac{T}{2}}.
	\end{equation*}
For simplicity, we let the prior density $\p{}{\upsilon}$ be a standard normal density. We set the true variance to $e^\upsilon =1$ and simulate $T=100$ data points from model (\ref{eq:ex_model}). The overall performance of Metropolis sampling in this model is presented in figure (\ref{fig:MH_sim}), where the integrated autocorrelation factor is estimated to be $\tau_{\widehat{\upsilon}}\approx 5.6$ using Sokal's adaptive truncated periodogram estimator \citep{Sokal1997}. When using the alternative implementation of Metropolis sampling given by (\ref{eq:refresh_V}-\ref{eq:ext_slice}), we found a reduction in autocorrelation factor to $\tau_{\widehat{\upsilon}}\approx 4.6$. A close-up look at the sample trace reveals numerous repeated samples, which is distinctive to Metropolis sampling with optimal acceptance rate $\widebar{\alpha}\approx 0.44$ when $n_\upsilon = 1$. We also note that this simple one--parameter modelling approach yields a very simplistic posterior predictive density for $\scripted{y}_t$, which is effectively a flat band of $90\%$--credible interval across all $t=1,2,\dots T$.
	\begin{figure}
	\center\includegraphics[width=0.64\linewidth]{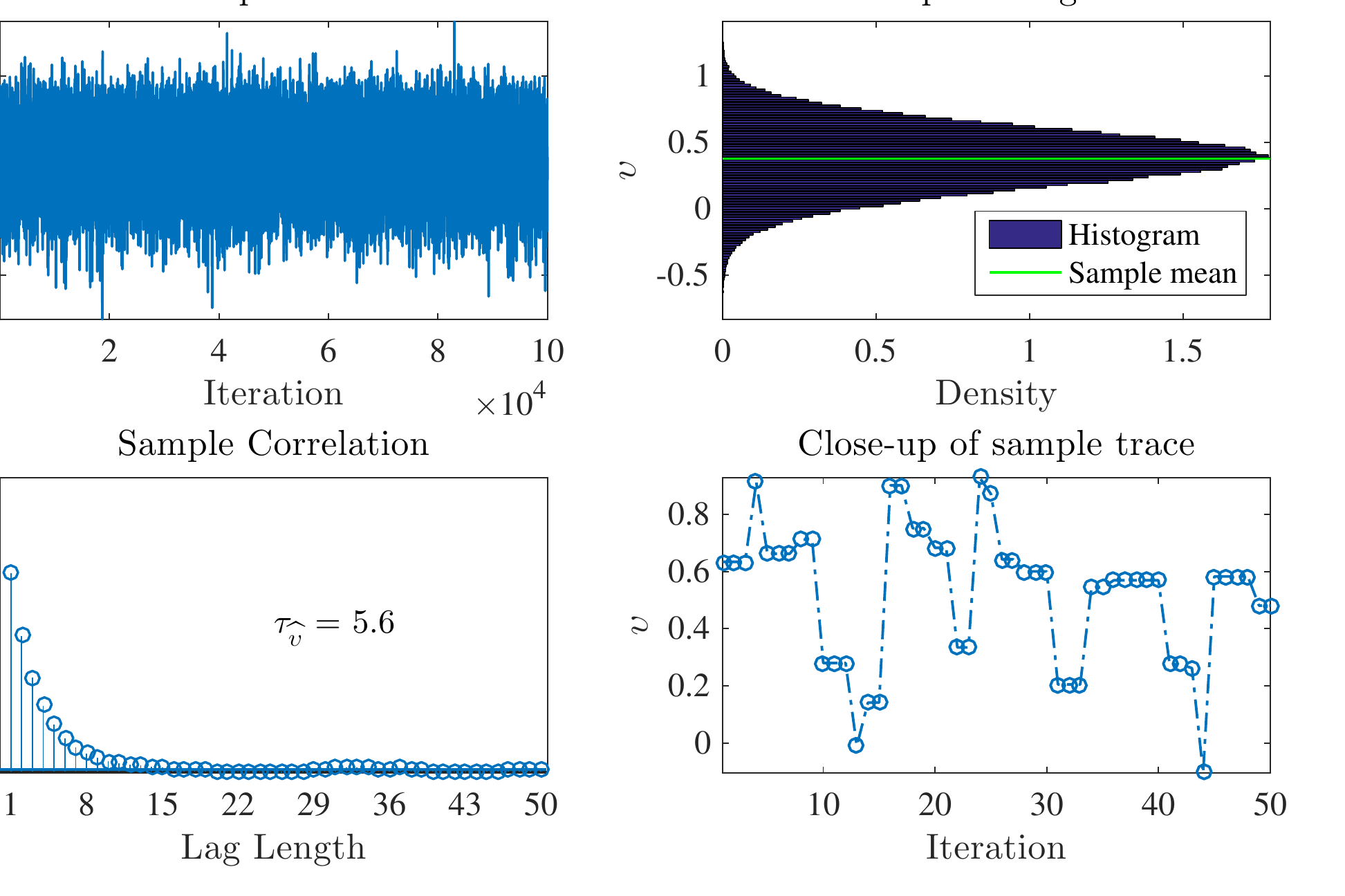}
	\includegraphics[width=0.34\linewidth]{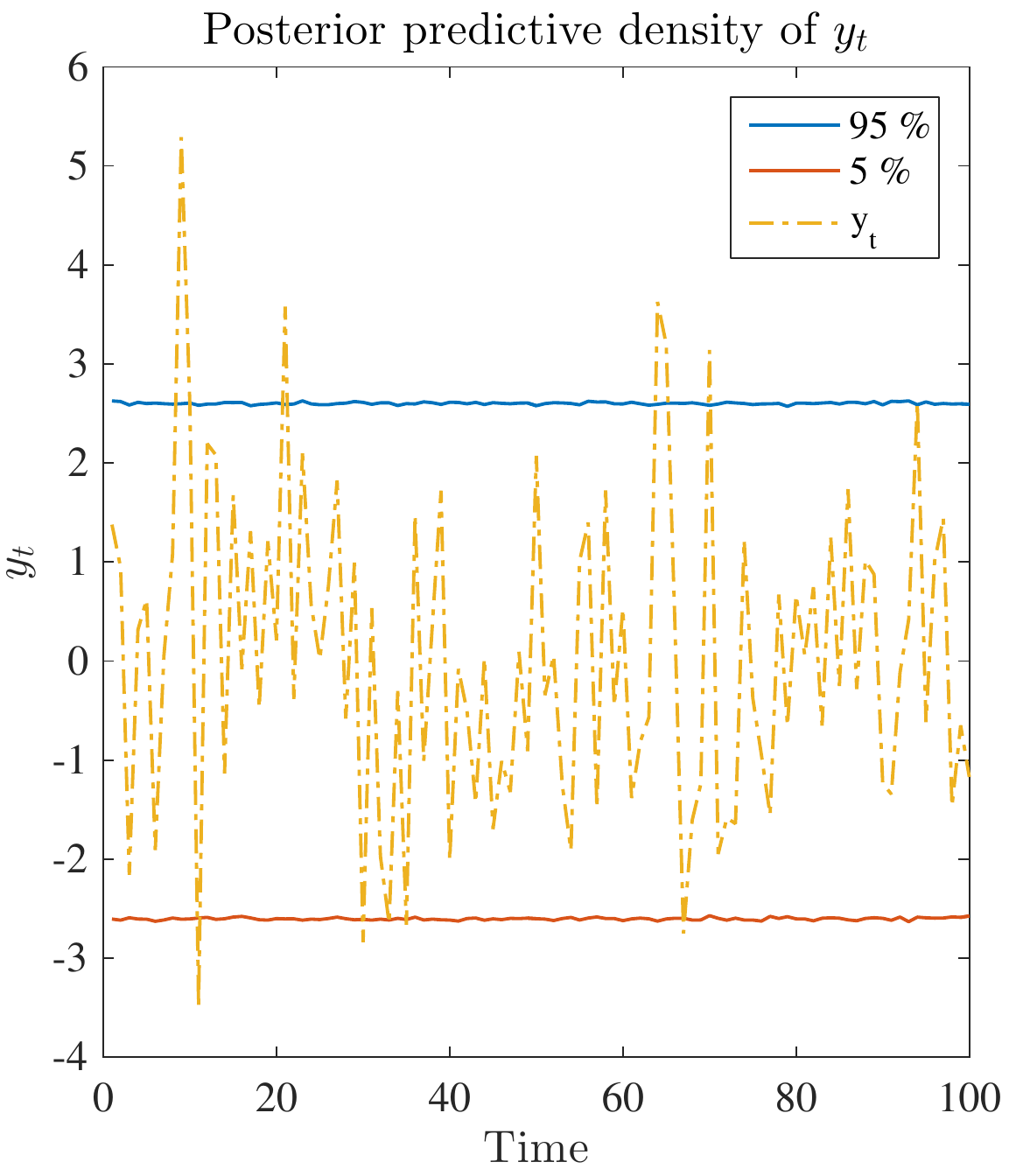}
  	\caption{Sampling $\upsilon_k\sim\p{}{\upsilon\mid Y_T}$ using Adaptive Metropolis from \citep{Andrieu2008}}
  	\label{fig:MH_sim}
	\end{figure}
\end{example}
\begin{example}\label{ex:Gibbs}
	The model (\ref{eq:ex_model}) can also be interpreted as a multi--parameters model with 
	\begin{equation*}
		\theta\defeq\mat{\upsilon & \scripted{x}_1 &\dots &\scripted{x}_T}\in\mathbb{R}^{T+1},
	\end{equation*} and the following  posterior density  
\begin{equation}\label{eq:ex_posterior}
	\pi(\theta)\defeq\p{}{\theta\mid Y_T} \propto\p{}{Y_T\mid\theta}\p{}{\theta}\propto \frac{\p{}{\upsilon}}{\exp(\upsilon)^\frac{T}{2}} \prod_{t=1}^T\exp\left(\frac{-\scripted{x}_t^2}{2}-\frac{(\scripted{y}_t-\scripted{x}_t)^2}{2e^\upsilon}\right)
\end{equation} For $T=100$, Metropolis sampling is no longer practical since the adaptive approximation of the covariance matrix $\Sigma_k$ becomes very erratic at high dimension. Instead, we need to perform the following Gibbs sampling scheme
\begin{align}
	\upsilon\sim\p{}{\upsilon\mid\theta\setminus\upsilon,Y_T} &\propto\frac{\p{}{\upsilon}}{\exp(\upsilon)^\frac{T}{2}} \prod_{t=1}^T\exp\left(-\frac{(\scripted{y}_t-\scripted{x}_t)^2}{2e^\upsilon}\right),\label{eq:a_0} \\
	\scripted{x}_t\sim\p{}{\scripted{x}_t\mid\theta\setminus\scripted{x}_t,Y_T} &\propto\exp\left(\frac{-\scripted{x}_t^2}{2}-\frac{(\scripted{y}_t-\scripted{x}_t)^2}{2e^\upsilon}\right) \text{ for } t= 1,2,\dots T,\label{eq:a_t}
\end{align} where (\ref{eq:a_0}) is a univariate density and can be solved by any Markov chain Monte Carlo algorithm, we used slice sampling \citep{Neal2003} in this instance. We can also sample exactly from the other $T$ conditional densities of the form (\ref{eq:a_t}) since they are products of 2 Gaussians which are themselves indeed Gaussian with variance $\sqrt{\frac{e^\upsilon}{e^\upsilon+1}}$ and mean $\scripted{y}_te^\upsilon$ respectively. We see in figure (\ref{fig:Gibbs_sim}) that this sampling approach produces virtually i.i.d. samples from $\p{}{\scripted{x}_t\mid\theta\setminus\scripted{x}_t,Y_T}$, e.g. $\tau_{\widehat{\scripted{x}_{60}}}\approx 1$, while the samples from $\p{}{\upsilon\mid\theta\setminus\upsilon,Y_T}$ are essentially as uncorrelated as the samples in example (\ref{ex:Metropolis}). Since the model now includes inference on the values of $\scripted{x }_t$, the posterior predictive density for $y_t$ is now much more realistic than the single parameter modelling approach in example (\ref{ex:Metropolis}).
\begin{figure}[hbt]
	\center
	\includegraphics[width=0.49\linewidth]{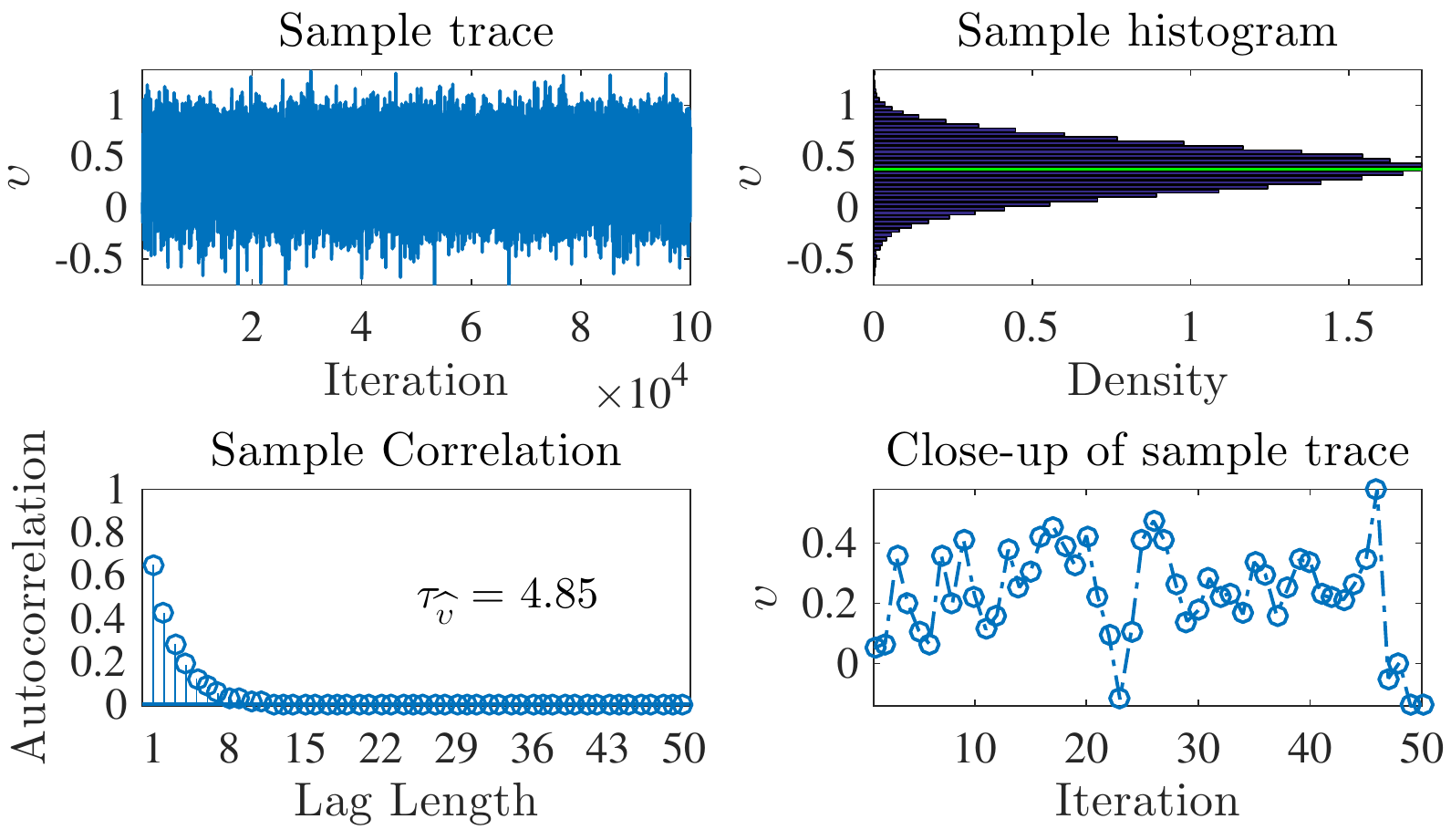}
	\includegraphics[width=0.49\linewidth]{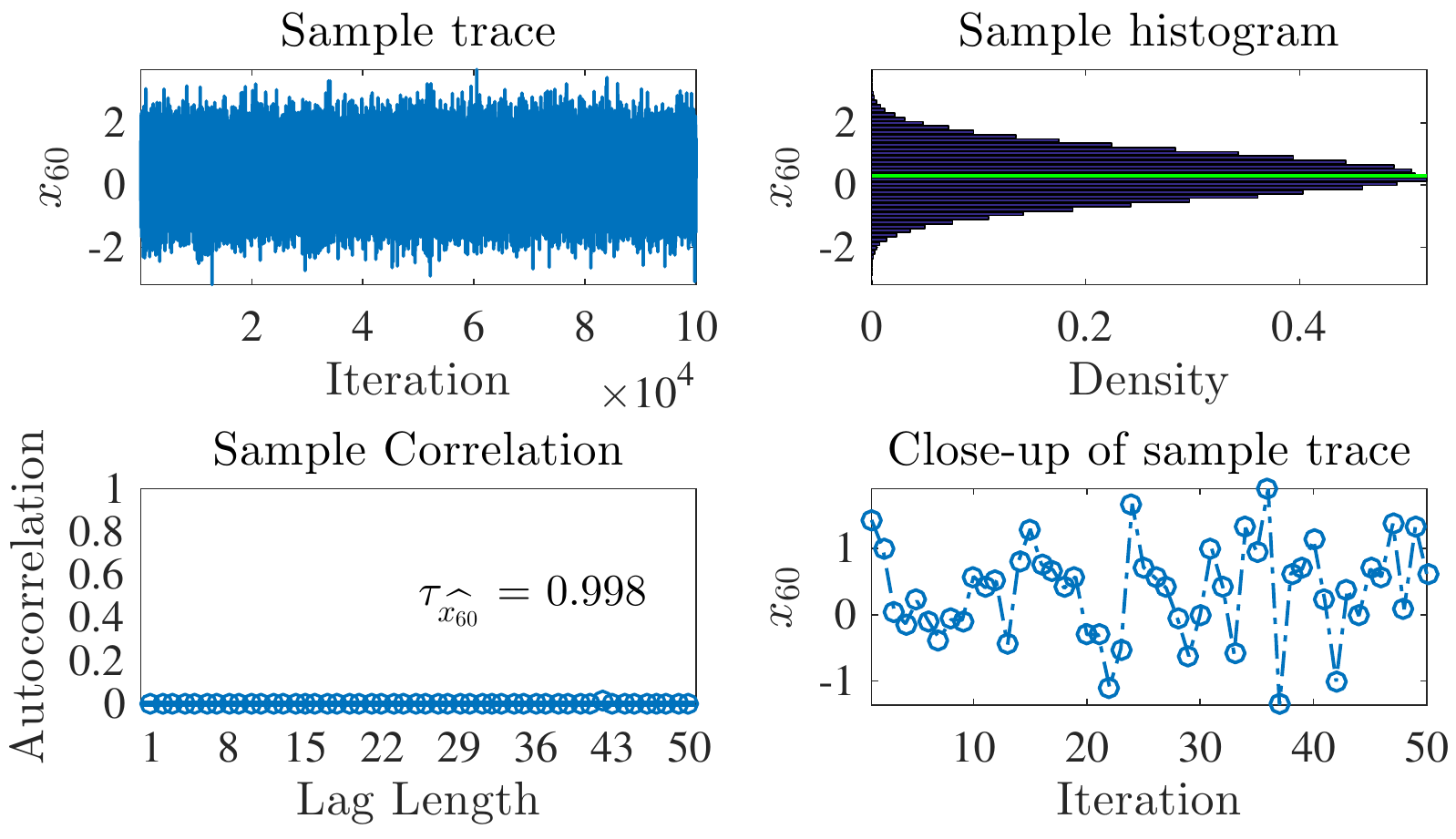}
	\includegraphics[width=\linewidth]{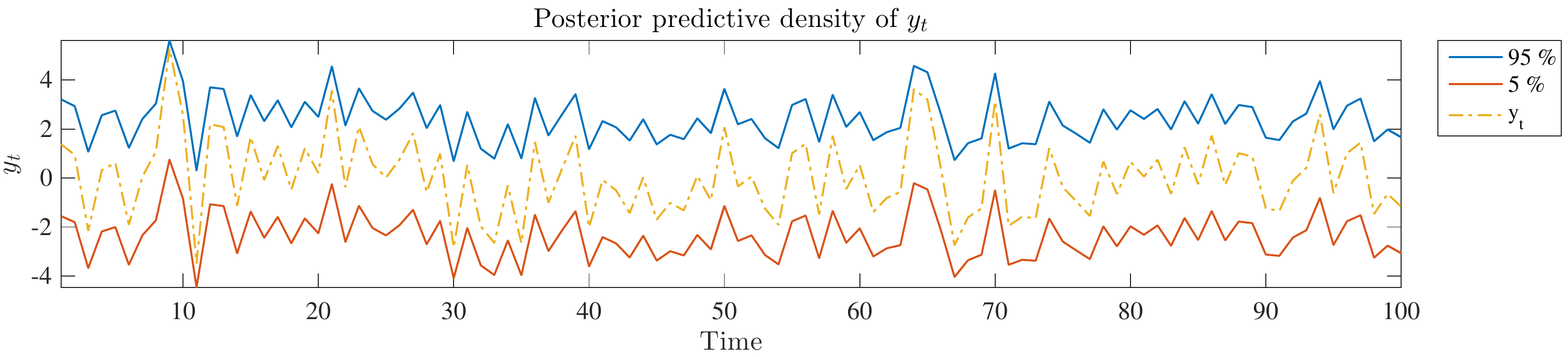}
  	\caption{Gibbs sampling for $\theta\defeq\SqrBrk{\upsilon\;  \scripted{x}_1\; \dots\; \scripted{x}_T}\sim\p{}{\theta\mid Y_T}$}
  	\label{fig:Gibbs_sim}
\end{figure}
\end{example}
\begin{example}\label{ex:aux_Gibbs}
Despite the fact that the Gibbs sampler is applicable to model (\ref{eq:ex_model}) without adding any extra auxiliary variables, we can employ the auxiliary variables method, just for a comparison, by introducing one auxiliary variable for each factor in the density (\ref{eq:ex_posterior}) as ${\mathbb{h}\defeq\mat{h_0 & h_1 & h_2 &\dots h_T} }$ with
\begin{equation*}
		h_t\mid\theta,Y_T \sim\mathcal{U}\left[0,\scripted{l}_t(\theta)\right],\quad\theta\defeq\mat{\upsilon & \scripted{x}_1 &\dots &\scripted{x}_T},
\end{equation*} where $\scripted{l}_0(\theta)\defeq \p{}{\upsilon}$ and for $t = 1,2,\dots T$
\begin{equation*}
	\scripted{l}_t(\theta)\defeq \frac{\exp\left(- \frac{\scripted{x}_t^2}{2}-\frac{(y_t-\scripted{x}_t)^2}{2e^\upsilon}\right)}{\sqrt{e^\upsilon}} = \frac{s_t}{\sigma}\exp\left(-\frac{(\scripted{x}_t - \mu_t)^2}{2\sigma^2} \right),
\end{equation*} where the derivation for the rightmost expression above can be found in \citep{Bromiley2014} with
\begin{equation}\label{eq:prod_2_Gaus}
\begin{aligned}
	\mu_t &\defeq y_t\sigma^2;\quad
	\sigma^2 &\defeq\frac{e^\upsilon}{e^\upsilon+1};\quad
	s_t &\defeq\frac{\exp\left(-\frac{y_t^2}{2(e^\upsilon+1)}\right)}{\sqrt{e^\upsilon+1}}.
\end{aligned}
\end{equation} Therefore, we can solve the conditions $h_t < \scripted{l}_t(\theta) $ for $\scripted{x}_t$ to results in
\begin{equation}\label{eq:Gibbs_x}
	\scripted{x}_t = \mu_t + u_t\sigma\sqrt{-2\log\left(\frac{\sigma h_t}{s_t}\right)};\quad -1< u_t < 1.
\end{equation}
Now the Gibbs sampling scheme to sample from $\p{}{\theta,\mathbb{h}\mid Y_T}$ can be completed with 
\begin{align}
	\scripted{x}_t\mid\theta\setminus\scripted{x}_t ,\mathbb{h},Y_T &\sim\mathcal{U}\left[\mu_t - \sigma\sqrt{-2\log\left(\frac{\sigma h_t}{s_t}\right)},\mu_t + \sigma\sqrt{-2\log\left(\frac{\sigma h_t}{s_t}\right)}\right],\label{eq:aux_Gibbs_x} \\
	\upsilon\mid\theta\setminus\upsilon,\mathbb{h},Y_T &\sim \prod_{t=0}^T\mathbb{1}\{h_t < \scripted{l}_t(\theta)\}\label{eq:aux_Gibbs_ups},
\end{align} where the boundary condition of (\ref{eq:aux_Gibbs_x}) is solved exactly by (\ref{eq:Gibbs_x}) while (\ref{eq:aux_Gibbs_ups}) has $T+1=101$ boundary conditions and is solved by slice sampling \citep{Neal2003}. The performance of this approach is given in figure (\ref{fig:aux_Gibbs}), where we can see the Markov chain exploring the $T=100$ dimensions of $\scripted{x}_t$ quite freely while significantly stronger sample correlation appears in the $\upsilon$--component, due to the $101$ concurrent constrains placed on the movement of $\upsilon$ in each iteration.
\begin{figure}[hbt]
\center
\includegraphics[width=0.49\linewidth]{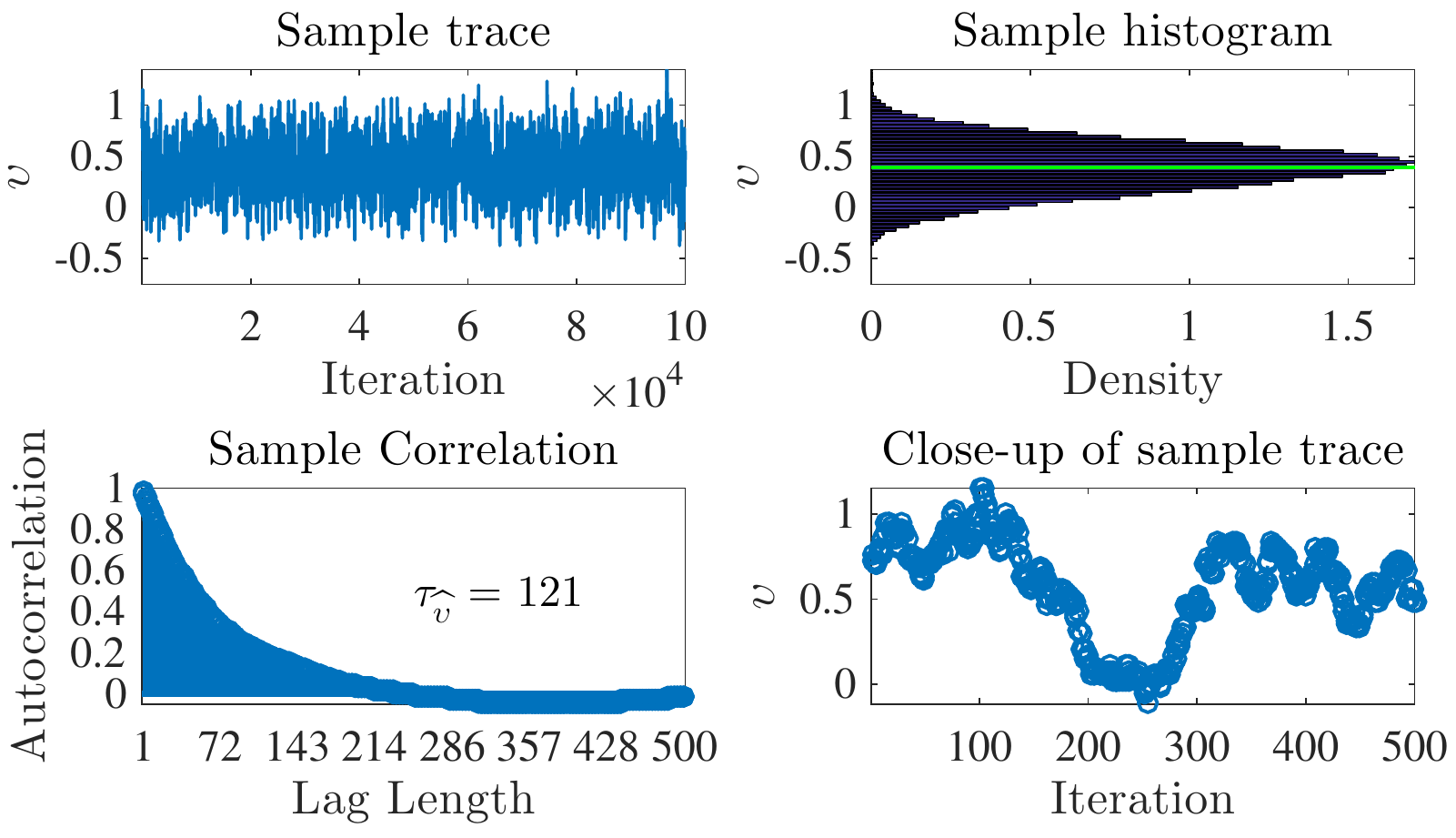}
\includegraphics[width=0.49\linewidth]{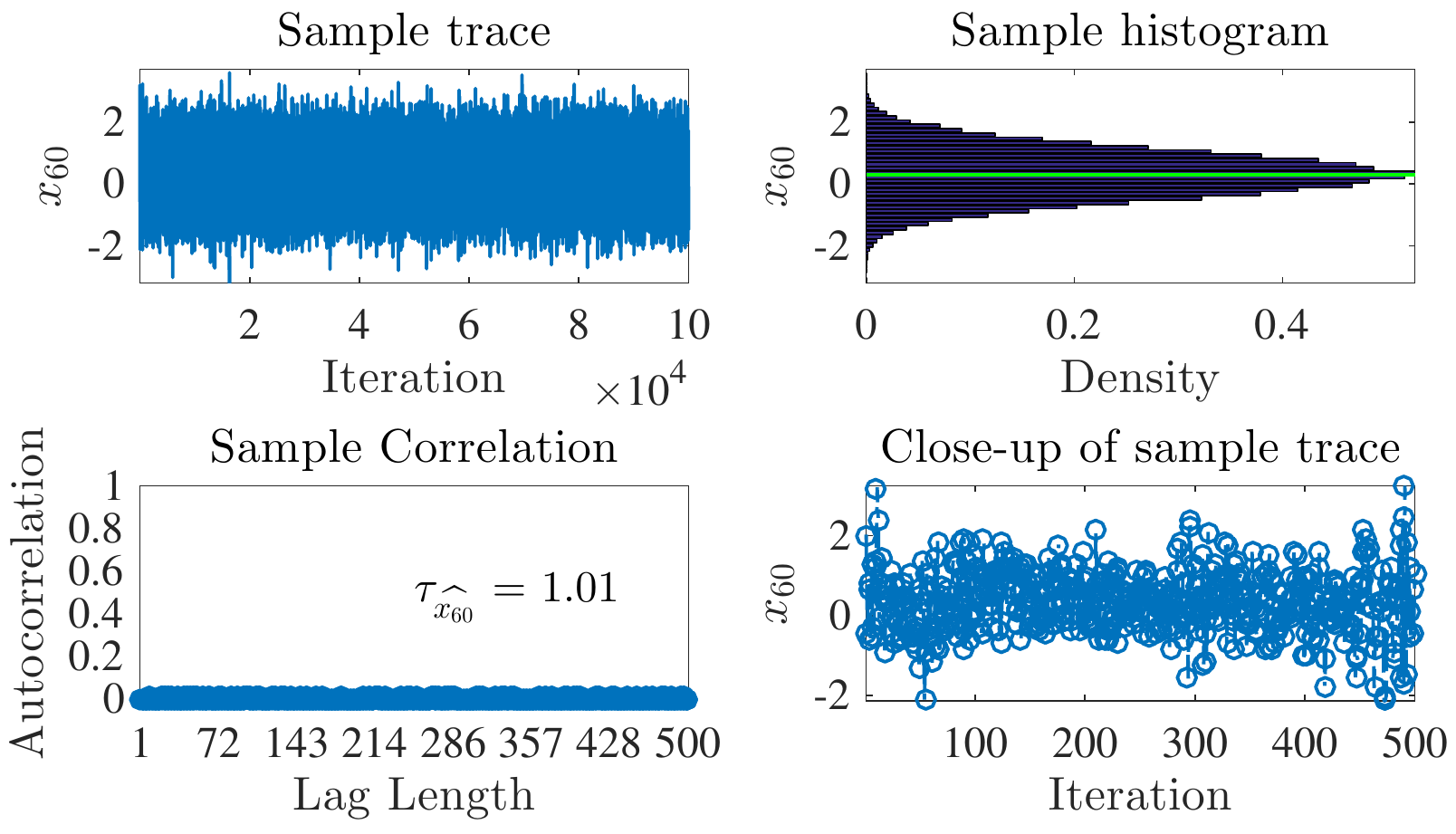}
\caption{Sampling for $\theta\defeq\SqrBrk{\upsilon\;  \scripted{x}_1\; \dots\; \scripted{x}_T}\sim\p{}{\theta\mid Y_T}$ using auxiliary variables method}
\label{fig:aux_Gibbs}
\end{figure}
\end{example}
\begin{example}\label{ex:slice}
	We have seen in example (\ref{ex:aux_Gibbs}) that introducing too many auxiliary variables could eventually induce excessive restriction on the movement of $\upsilon$. To avoid this phenomenon, we can introduce only one auxiliary variable as
\begin{equation*}
	\scripted{h}\mid\theta,Y_T \sim \mathcal{U}\left[0,\pi(\theta)\right];\quad\theta\defeq\SqrBrk{\upsilon\;  \scripted{x}_1\; \dots\; \scripted{x}_T};\quad\pi(\theta)\defeq\p{}{\theta\mid Y_T},
\end{equation*} such that the joint density of $(\theta,\scripted{h})$ given $Y_T$ becomes uniform. The conditional densities for $\upsilon$ and each $\scripted{x}_t, t = 1,2,\dots T,$ also become uniform as follows
\begin{equation*}
\begin{aligned}
	\upsilon\mid\theta\setminus\upsilon,\scripted{h}, Y_T &\sim\mathcal{U}[\scripted{S}_\upsilon];\quad\scripted{S}_\upsilon\defeq\left\lbrace\upsilon : \scripted{h} \le\pi(\theta)\right\rbrace ,\\
	\scripted{x}_t\mid\theta\setminus\scripted{x}_t,\scripted{h}, Y_T &\sim\mathcal{U}[\scripted{S}_{\scripted{x}_t}];\quad\scripted{S}_{\scripted{x}_t}\defeq\left\lbrace\scripted{x}_t : h\le\pi(\theta) \right\rbrace.
\end{aligned}
\end{equation*}
We note that $\scripted{S}_\upsilon$ and $\scripted{S}_{\scripted{x}_t}$ are univariate subsets of the $n_\theta$--variates slice $\scripted{S}$ in (\ref{eq:slice}). In this case the boundaries of $\scripted{S}_{\scripted{x}_t}$ can also be analytically derived in a similar manner as with (\ref{eq:Gibbs_x}) as follows
\begin{equation*}
	h_t\defeq\frac{\scripted{h}}{\frac{\p{}{\upsilon}}{e^{\frac{\upsilon(T-1)}{2}}}\exp\left(\sum\limits_{i\in\{1:T\}\setminus t}-\frac{\scripted{x}_i^2}{2}-\frac{(y_i-\scripted{x}_i)^2}{2e^\upsilon} \right)}\le\frac{1}{e^{\upsilon/2}}\exp\left(-\frac{\scripted{x}_t^2}{2} -\frac{(y_i-\scripted{x}_t)^2}{2e^\upsilon}  \right),
\end{equation*} which according to (\ref{eq:ex_posterior}) and (\ref{eq:prod_2_Gaus})--(\ref{eq:Gibbs_x}) will result in all of the conditional densities for $\scripted{x}_t$ having exact solution as follows
\begin{equation*}
	\scripted{x}_t\mid\theta\setminus\scripted{x}_t,\scripted{h}, Y_T \sim\mathcal{U}\left[\mu_t - \sigma\sqrt{-2\log\left(\frac{\sigma h_t}{s_t}\right)},\mu_t + \sigma\sqrt{-2\log\left(\frac{\sigma h_t}{s_t}\right)}\right].
\end{equation*}
Furthermore, the conditional sampling for $\upsilon$ only have to satisfy one single constrain, which results in lower sample correlation than example (\ref{ex:aux_Gibbs}), as seen in figure (\ref{fig:slice}).
\begin{figure}[hbt]
\center
\includegraphics[width=0.49\linewidth]{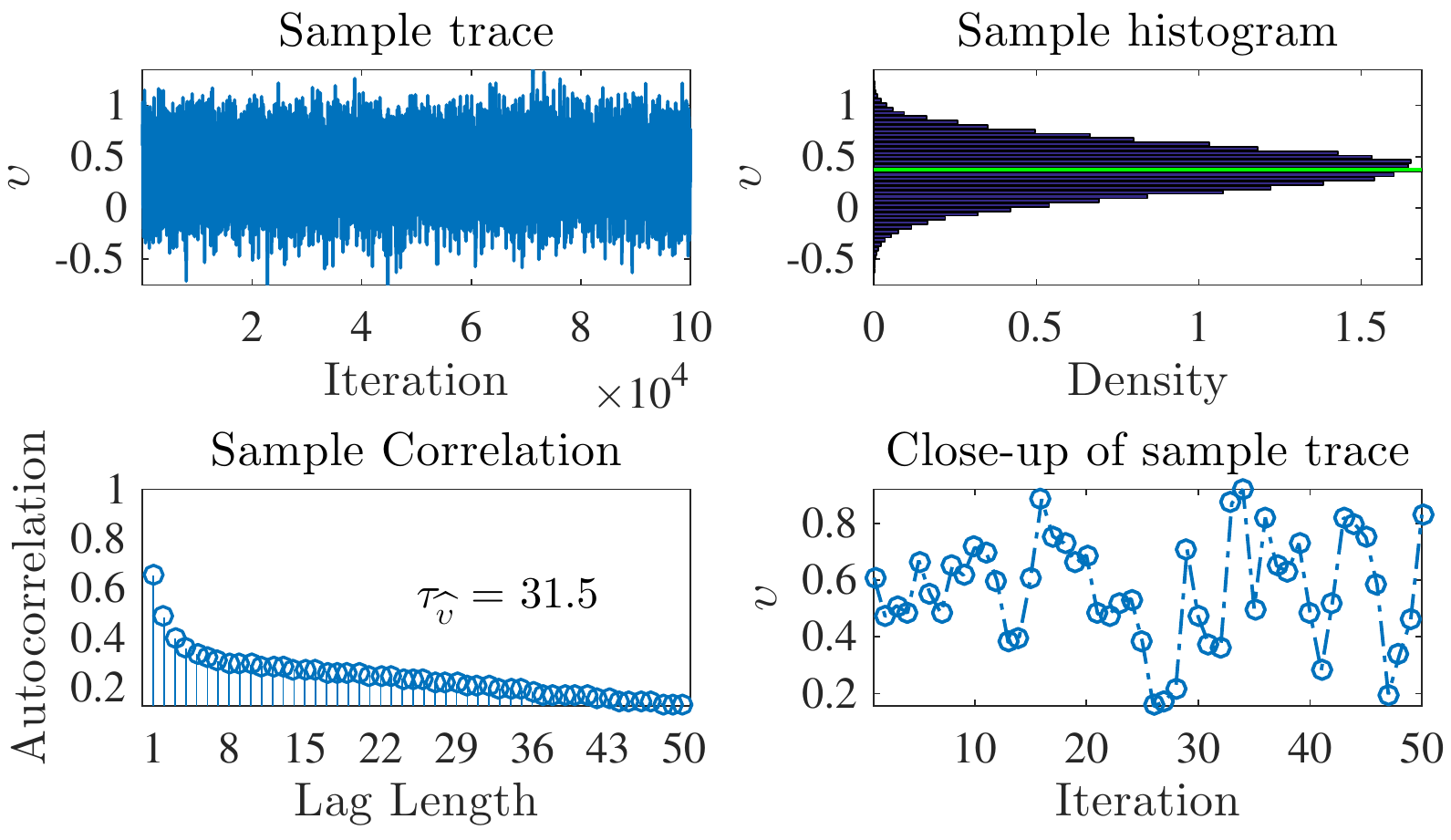}
\includegraphics[width=0.49\linewidth]{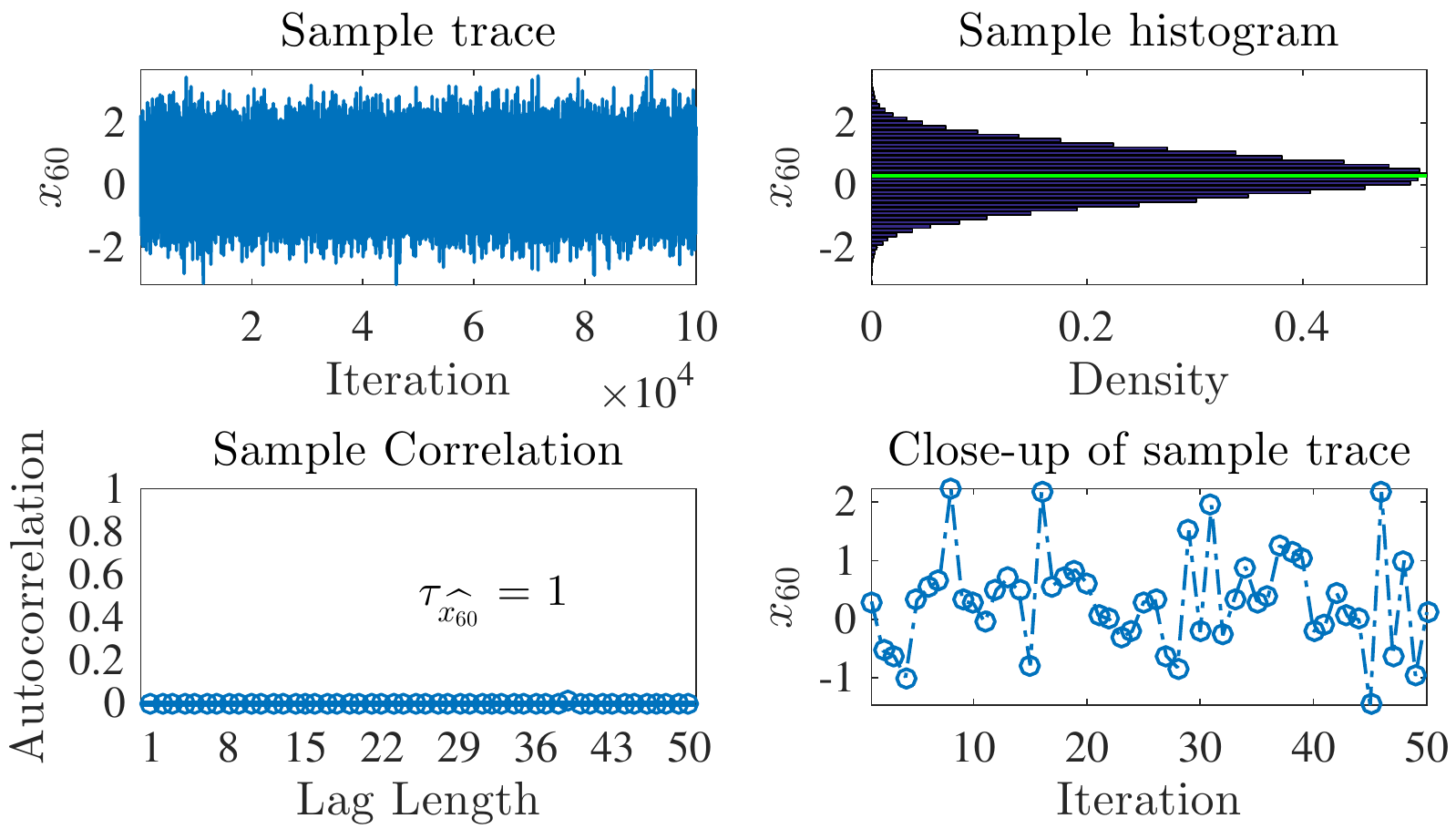}
\caption{Sampling for $\theta\defeq\SqrBrk{\upsilon\;  \scripted{x}_1\; \dots\; \scripted{x}_T}\sim\p{}{\theta\mid Y_T}$ using single auxiliary variable method, i.e. slice sampling}
\label{fig:slice}
\end{figure}
\end{example}

\bibliographystyle{biometrika}
\bibliography{paper-ref}
\end{document}